%% file: massinflation.220613.arxiv.tex
\documentclass[10pt, reqno]{amsart}
\usepackage{amsthm,amsfonts,amssymb,euscript,cite,slashed,mathabx,mathrsfs}
\usepackage{graphics}
\usepackage[hidelinks]{hyperref}
\usepackage{xcolor} 
\usepackage{hyperref}
\usepackage{graphicx}
\usepackage{color}
\usepackage{transparent}

\usepackage{enumitem}
\newlist{primenumerate}{enumerate}{1}
\setlist[primenumerate,1]{label={(\arabic*$'$)}}

\usepackage{fullpage} 

\definecolor{green}{rgb}{0,0.8,0} 

\newtheorem{theorem}{Theorem}[section]
\newtheorem{corollary}[theorem]{Corollary}
\newtheorem{lemma}[theorem]{Lemma}
\newtheorem{proposition}[theorem]{Proposition}
\theoremstyle{definition}
\newtheorem{definition}[theorem]{Definition}

\theoremstyle{remark}
\newtheorem{remark}[theorem]{Remark}
\newtheorem{conjecture}[theorem]{Conjecture}
\numberwithin{equation}{section}

\newcommand{\abs}[1]{\vert#1\vert}

\newcommand{\set}[1]{\{#1\}}

\newcommand{\ep}{\epsilon}
\def\beaa{\begin{eqnarray*}}
\def\eeaa{\end{eqnarray*}}
\def\bea{\begin{eqnarray}}
\def\eea{\end{eqnarray}}
\def\be{\begin{equation}}
\def\ee{\end{equation}}
\def\ls{\lesssim}

\newcommand{\ud}{\mathrm{d}}
\newcommand{\rd}{\partial}

\newcommand{\alp}{\alpha}
\newcommand{\bt}{\beta}
\newcommand{\gmm}{\gamma}

\newcommand{\kpp}{\kappa}

\newcommand{\sgm}{\sigma}

\newcommand{\tht}{\theta}

\newcommand{\om}{\omega}
\newcommand{\Omg}{\Omega}

\newcommand{\bfe}{{\bf e}}

\newcommand{\calC}{\mathcal C}

\newcommand{\calH}{\mathcal H}

\newcommand{\calM}{\mathcal M}

\newcommand{\f}{\frac}

\newcommand{\T}{\mathbb T}

\newcommand{\CH}{\calC \calH^{+}}		
\newcommand{\EH}{\calH^{+}}
\newcommand{\pfstep}[1]{\vspace{.5em} {\it #1.}}
\newcommand{\vol}{\mathrm{vol}}

\newcommand{\Int}{\mathscr{B}}
\newcommand{\Ext}{\mathscr{E}}

\begin{document}

\title[]{A scattering theory approach to Cauchy horizon instability\\ and applications to mass inflation}

\author{Jonathan Luk}
\address{Department of Mathematics, Stanford University, 450~Serra~Mall~Building~380,~Stanford~CA~94305-2125,~USA}
\email{jluk@stanford.edu}

\author{Sung-Jin Oh}%
\address{Department of Mathematics, UC Berkeley, Berkeley, CA 94720, USA and KIAS, Seoul, Korea 02455}%
\email{sjoh@math.berkeley.edu}%

\author{Yakov Shlapentokh-Rothman}
\address{Department of Mathematics, University of Toronto, 40 St. George Street, Toronto, Ontario, Canada M5S 2E4 and Department of Mathematics and Computational Sciences, 3359 Mississauga Road, Mississauga, Ontario, Canada L5L 1C6}
\email{yakovsr@math.toronto.edu}


\begin{abstract}
Motivated by the strong cosmic censorship conjecture, we study the linear scalar wave equation in the interior of subextremal strictly charged Reissner--Nordstr\"om black holes by analyzing a suitably-defined ``scattering map'' at $0$ frequency. The method can already be demonstrated in the case of spherically symmetric scalar waves on Reissner--Nordstr\"om: we show that assuming suitable ($L^2$-averaged) upper and lower bounds on the event horizon, one can prove ($L^2$-averaged) polynomial lower bound for the solution
\begin{enumerate}
\item on any radial null hypersurface transversally intersecting the Cauchy horizon, and
\item along the Cauchy horizon towards timelike infinity.
\end{enumerate}
Taken together with known results regarding solutions to the wave equation in the exterior, (1) above in particular provides yet another proof of the linear instability of the Reissner--Nordstr\"om Cauchy horizon. As an application of (2) above, we prove a conditional mass inflation result for a nonlinear system, namely, the Einstein--Maxwell--(real)--scalar field system in spherical symmetry. For this model, it is known that for a generic class of Cauchy data $\mathcal G$, the maximal globally hyperbolic future developments are $C^2$-future-inextendible. We prove that if a (conjectural) improved decay result holds in the exterior region, then for the maximal globally hyperbolic developments arising from initial data in $\mathcal G$, the Hawking mass blows up identically on the Cauchy horizon. 
\end{abstract}

\maketitle

\section{Introduction}\label{sec:intro}

In this paper, we consider the linear scalar wave equation (where $\Box_g$ denotes the Laplace--Beltrami operator)
\begin{equation}\label{wave.eqn}
\Box_{g}\phi = 0.
\end{equation}
in the interior of Reissner--Nordstr\"om (with $0<|{\bf e}|<M$) black holes. Hidden in the interior of these black holes are the so-called Cauchy horizons, whose stability and instability properties are of fundamental importance due to their intimate connections with the strong cosmic censorship conjecture and the problem of determinism; see further discussions in Section~\ref{sec:SCC}.

The equation \eqref{wave.eqn} in the black hole interior region, for both the Reissner--Nordstr\"om and the Kerr cases are rather well-understood. Let us just focus on the following definitive $C^0$-stability and non-degenerate energy-instability results. (We refer the readers for instance to \cite{D2, DafShl, CH, GSNS, Gleeson, KSR, McN.stab, McN, Sbi.2} and the references therein for related results.)
\begin{itemize}
\item (Stability \cite{Fra, Franzen2, Hintz, LS}) On both Reissner--Nordstr\"om (with $0<|{\bf e}|<M$) and Kerr (with $0<|a|<M$), solutions $\phi$ arising from smooth and compactly supported Cauchy data on $\Sigma_0$ (see Figure~\ref{fig:main.thm}) are uniformly bounded up to the Cauchy horizon and are continuously extendible to the Cauchy horizon. In fact, $|\phi|$ decays along the Cauchy horizon towards timelike infinity.
\item (Instability \cite{D2,LO.instab, LS}) On both Reissner--Nordstr\"om (with $0<|{\bf e}|<M$) and Kerr (with $0<|a|<M$), if an $L^2$-averaged lower bound for the derivative for $\phi$ holds on the event horizon, then $\phi$ has infinite non-degenerate energy on a null hypersurface intersecting the Cauchy horizon transversely. In particular, the derivatives of $\phi$ blow up at the Cauchy horizon. 

Moreover, using the results of \cite{AAG,HintzPriceLaw,LO.instab}, this assumed $L^2$-averaged lower bound on the event horizon is proven to be satisfied by solutions arising from \emph{generic} smooth and compactly supported Cauchy data on $\Sigma_0$.
\end{itemize}

\begin{figure}[h]
\begin{center}
\def\svgwidth{200px}
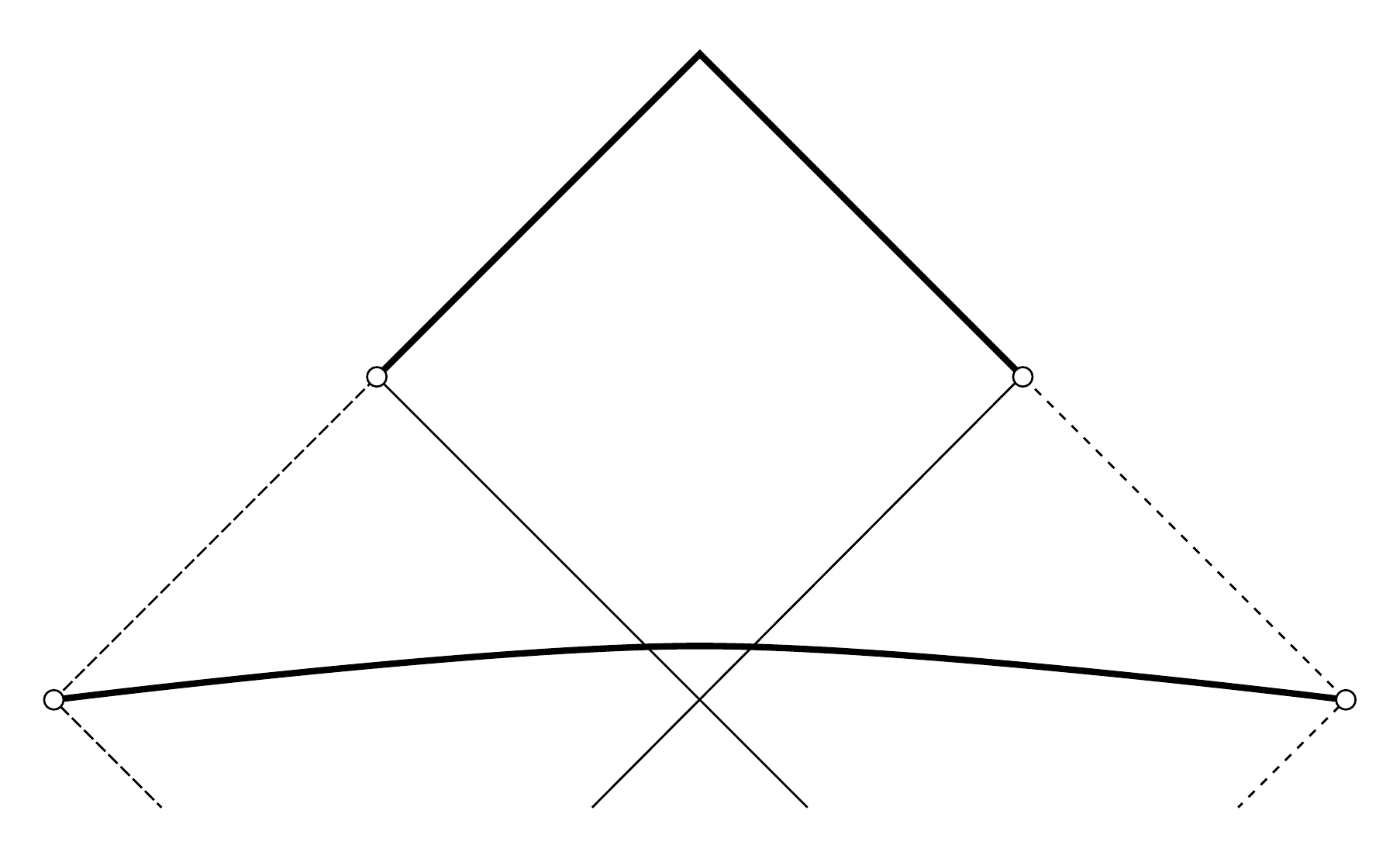 
\caption{} \label{fig:main.thm}
\end{center}
\end{figure}

Even though the stability results for \eqref{wave.eqn} are necessarily quite weak\footnote{For instance, in terms of (isotropic) H\"older and Sobolev spaces, it can be deduced in the Reissner--Nordstr\"om case using the result in \cite{AAG, D2, Gleeson} that a generic solution is neither in any H\"older $C^{\alpha}$ space for $\alp>0$ nor in any Sobolev $W^{1,p}$ space for $p>1$.} because of the instability results, they are very robust. In particular, a slight modification of the proof gives similar stability results for very general systems of wave equations with first- or zeroth-order terms. The robustness of the linear stability, together with the remarkable nonlinear structure of the Einstein vacuum equations in the double null foliation gauge, result in the proof of the \emph{nonlinear} $C^0$ stability of the Kerr Cauchy horizon for the Einstein vacuum equations without any symmetry assumptions \cite{DL}.

On the other hand, the proofs of linear (non-degenerate) energy instability for \eqref{wave.eqn} are much less robust\footnote{In fact it may be possible that not only the proofs fail, but the instability result itself is false when fine-tuned lower order terms are added to \eqref{wave.eqn}. This, however, remains an open problem.}. The known proofs \cite{LO.instab, LS} of these instabilities rely heavily on the conservation law associated to the Killing field $T$ (defined in \eqref{eq:def.Killing}), which no longer holds if one adds lower order terms to \eqref{wave.eqn}. For this reason, linear instability for the full system of linear gravitational perturbations remains an open problem. At the moment, nonlinear instability results have only been obtained for spherically symmetric models; see \cite{D2, LO.interior, VDM, VDM4} and Section~\ref{sec:SSModel}. Nevertheless, one expects that if one can prove an instability result for the full system of linear gravitational perturbations, then the techniques in \cite{DL} could in principle be sufficient to control all the nonlinear error terms and to upgrade the linear instability result to a nonlinear result.
	
\textbf{The first goal of this paper is to revisit the linear non-degenerate energy instability result} with a proof which is potentially generalizable to studying instabilities for linear gravitational perturbations. In fact, the ideas we present in this paper have already been taken up by Sbierski, who proved a linear instability result for the Teukolsky equation in a forthcoming work \cite{Sbierski.Teukolsky}.

The perspective of this paper is to study \eqref{wave.eqn} by introducing a scattering problem, where the past and future ``scattering states'' are the restriction of the solution to the wave equation on the event horizon and the Cauchy horizon respectively. (Note that such a point of view for the black hole interior is not new, and has been used in \cite{DafShl,McN,KSR}. See also \cite{KehleAdS,KehleVdM} in different settings. The results in \cite{KSR} are especially relevant to our paper.) In particular, we analyze the corresponding ``transmission coefficient'' and ``reflection coefficient'' at $0$ $t$-frequency, and use that information to study the instability property of the Cauchy horizon. Our approach is in part inspired by the work \cite{GSNS}, which already recognized the important role of the transmission and reflection coefficients at zero frequency, at least for a class of data with exact polynomial tail.

It turns out that a slight modification of our new linear instability proof also gives a lower bound of the scalar field along the Cauchy horizon towards timelike infinity. \textbf{The second goal of this paper is to use this lower bound and apply it to the mass inflation problem} for the Einstein--Maxwell--(real) scalar field system in spherical symmetry. Since the mass inflation problem requires a longer discussion, we will explain this application later in Section~\ref{sec:SCC}; see Theorem~\ref{thm:intro.MI.2}. The key point here is to demonstrate that our linear result, which is established with a fundamentally linear proof, can be easily applied to a nonlinear setting via a perturbative argument.

It turns out that all the main ideas can be demonstrated already in the case of solutions to \eqref{wave.eqn} on Reissner--Nordstr\"om with spherically symmetric data. We will therefore focus on that case from now on. See Remark~\ref{rmk:more.general} for a discussion of possible generalizations. The following is an informal version of the main theorem for \eqref{wave.eqn} on Reissner--Nordstr\"om; a precise version will be given in Theorem~\ref{thm:main} and Corollary~\ref{cor:main}. (Again, see Theorem~\ref{thm:intro.MI.2} for our main theorem on mass inflation.)

\begin{theorem}\label{thm:main.intro}
Let $\phi$ be a solution to \eqref{wave.eqn} on a fixed Reissner--Nordstr\"om spacetime (with $0<|{\bf e}|<M$) with smooth and spherically symmetric data.
Assume the following two conditions along the event horizon:
\begin{enumerate}
\item $T\phi$ obeys $L^2$-averaged polynomial upper and lower bounds, and
\item $T^2\phi$ obeys a \underline{better} (compared to $T\phi$) $L^2$-averaged polynomial upper bound (see precise assumptions in Theorem~\ref{thm:main}).
\end{enumerate}
Then, the following both hold:
\begin{enumerate}
\item Along any outgoing radial null hypersurface transversally intersecting the Cauchy horizon, the non-degenerate energy is infinite.
\item Along the Cauchy horizon, $T\phi$ obeys an $L^2$-averaged polynomial lower bound towards timelike infinity. 
\end{enumerate}
\end{theorem}

A few remarks are in order.

\begin{remark}[Comparison with known results]
The first result in Theorem~\ref{thm:main.intro}, which is an instability result, is \underline{not} new, see \cite{D2,LO.instab}. Moreover, a similar instability result is also known for Kerr spacetimes \cite{LS}. Our result is in fact slightly weaker than the known results in \cite{LO.instab, LS}, but our main concern is the introduction of a new method that is based on the phase space analysis of the scattering map near the zero $t$-frequency.

It should also be mentioned that scattering theory arguments has been used to show variations of this instability results, see \cite{CH, DafShl, KSR, McN}.
\end{remark}

\begin{remark}[Relation to mass inflation]
The second result in Theorem~\ref{thm:main.intro} (in contrast to the first result) is not directly related to the instability of the Cauchy horizon. It does, however, show that the decay result along the Cauchy horizon of \cite{Hintz} cannot be improved much further. Perhaps surprisingly, it has a nonlinear application to the problem of mass inflation for the Einstein--Maxwell--(real) scalar field system with two-ended asymptotically flat initial data in spherical symmetry; see Section~\ref{sec:mass.inflation}.
\end{remark}

\begin{remark}[More general settings]\label{rmk:more.general}
At least the instability part of Theorem~\ref{thm:main.intro} can in principle be generalized to many different settings, for instance, for higher angular modes or for the wave equation with a potential in Reissner--Nordstr\"om, or even for the wave equation on rotating Kerr backgrounds with a fixed Carter mode. Indeed,  to carry out the argument of Theorem~\ref{thm:main.intro} we need two main ingredients: (1) a stability result for solutions to the linear wave equation, and (2) an explicit computation showing that the transmission coefficient at zero frequency does not vanish. In all of the more general settings we mentioned above, the ingredient (1) is known (or at least follows from known techniques, see \cite{Hintz,LS}), while the ingredient (2) follows easily from a $T$-conservation law (see Remarks~\ref{rmk:grav.pert} and \ref{rmk:use.conservation.for.lower.bd}).
\end{remark}

\begin{remark}[Relevance to gravitational perturbations]\label{rmk:grav.pert}
We mentioned above that the previous instability results of \cite{LO.instab} and \cite{LS} both rely on the $T$-conservation law in a fundamental way (despite the fact that the two proofs are very different). In our setting, the $T$-conservation law could also be used to bound the zero-frequency transmission coefficient away from $0$; see part (3) of Proposition~\ref{prop:basic.T.R}. On the other hand, our method in principle does not require the $T$-conservation law as one can alternatively explicitly \emph{compute} the transmission coefficient at zero frequency, just as what we do for the reflection coefficient at zero frequency. For this reason, our method is relevant in settings where an analogue of the $T$-conservation law is not available, e.g.~in the case of gravitational perturbations of the Kerr interior;\footnote{In certain settings the Teukolsky-Starobinsky identities may be used as a replacement for the \emph{global} application of a T-conservation law, e.g., as an identity linking fluxes of suitable solutions to the Teukolsky equation along the Cauchy horizons and the event horizons. However, these identities cannot be localized in physical space in a straightforward fashion, and thus it does not allow for an immediate adaption of the arguments from \cite{LO.instab} and \cite{LS}.}  see \cite{Sbierski.Teukolsky}.
\end{remark}

In the remainder of the introduction, we will give a brief discussion of the strong cosmic censorship conjecture, which in particular serves as a motivation for the instability problem that we discuss in this paper. We will then turn to a discussion of the Einstein--Maxwell--(real) scalar field system in spherical symmetry and explain our application of Theorem~\ref{thm:main.intro} in that context.

\subsection{Background: strong cosmic censorship conjecture}\label{sec:SCC}

The study of the wave equation in the interior of black holes is motivated by the \emph{strong cosmic censorship conjecture} first proposed by Penrose \cite{PenroseSCCrefer}. In this subsection, we briefly review some mathematical progress on this conjecture, which motivates the results in the present paper. For a more detailed discussion of the strong cosmic censorship conjecture, we refer the reader to \cite{DL}.

A modern formulation of Penrose's strong cosmic censorship conjecture can be given as follows:
\begin{conjecture}[Strong cosmic censorship]\label{conj:SCC}
For generic asymptotically flat (or compact) vacuum initial data, the maximal Cauchy development is inextendible as a suitably regular Lorentzian manifold.
\end{conjecture}
This conjecture should be thought of as a conjecture on global uniqueness for the Einstein's equation. However, it is well-known that in the explicit Reissner--Nordstr\"om and Kerr black holes, there are Cauchy horizons beyond which the solution extends smoothly. In particular, this conjecture implies that the Cauchy horizons inside the Reissner--Nordstr\"om and Kerr black holes are \emph{unstable} in a suitable sense under small perturbations. 

The original expectation of the instability of the Cauchy horizons was modelled on the Schwarzschild solution, where instead of having a smooth Cauchy horizon, the black hole interior is singular and inextendible as a $C^0$ Lorentzian manifold \cite{Sbie.C0}. However, as has been recently proven \cite{DL}, it turns out that this expected instability is more subtle: the Cauchy horizons are in fact $C^0$ stable, and one can only expect the higher derivatives to blow up. This already manifests itself in the behavior of solutions to the linear equation \eqref{wave.eqn} (recall discussions the the beginning of Section~\ref{sec:intro}), but remarkably also holds in nonlinear settings.

Below, we will further explain the stability and instability issues in nonlinear settings. First, in Section~\ref{sec:C0.Kerr}, we discuss further the recent work on $C^0$-stability of the Kerr Cauchy horizon \cite{DL}. Then, in Section~\ref{sec:SSModel}, we discuss a spherically symmetric model, which is simpler and such that Conjecture~\ref{conj:SCC} is essentially settled. There remains, however, the problem of \emph{mass inflation} which is unresolved in the setting. Our Theorem~\ref{thm:main} turns out to give a conditional result in this regard. This will be explained in Section~\ref{sec:mass.inflation}.

\subsubsection{$C^0$-stability of the Kerr Cauchy horizon}\label{sec:C0.Kerr}
As already discussed above, the Kerr Cauchy horizon has recently been proven to be $C^0$ stable:

\begin{theorem}[Dafermos--Luk \cite{DL}]\label{thm:DL}
Consider general vacuum initial data corresponding to the expected induced geometry of a
dynamical black hole settling down to Kerr (with parameters $0 < |a| < M$) on a suitable spacelike hypersurface
$\Sigma_0$ in the black hole interior. Then the maximal future development spacetime corresponding to $\Sigma_0$
is globally covered by a double null foliation and has a non-trivial Cauchy horizon across which the
metric is continuously extendible.
\end{theorem}

Even with the above theorem, however, it is not known whether the Cauchy horizon is singular in any sense. The following conjecture remains an important open problem.

\begin{conjecture}\label{conj:instab}
For a generic subset of initial data as in Theorem~\ref{thm:DL}, the maximal Cauchy development is inextendible as a Lorentzian manifold with
continuous metric and Christoffel symbols locally square integrable.
\end{conjecture}

This conjecture motivates a better understanding of the linear instability. In particular, since for linear gravitational perturbations of Kerr, there is no obvious analogue of the $T$-conservation law, it is desirable to obtain a proof which does not rely in principle on such a conservation law. This motivates the considerations of the present paper.

\subsubsection{Strong cosmic censorship for spherically symmetric models}\label{sec:SSModel} From now on, we discuss Conjecture~\ref{conj:SCC} in spherical symmetry, focusing on the spherically symmetric Einstein--Maxwell--(real) scalar field system. We will in particular explain the background for the mass inflation problem in this context, and discuss an application of our Theorem~\ref{thm:main.intro} to this problem. In spherical symmetry, the problem becomes simpler than that in Section~\ref{sec:C0.Kerr}, we have a much more complete picture. More precisely, as we will describe below, not only an analogue of Theorem~\ref{thm:DL} is known, but moreover, (1) a \emph{global} $C^0$-stability result --- one that the initial data are posed on an asymptotically flat Cauchy hypersurface --- has been proven, and (2) the analogue of Conjecture~\ref{conj:instab} is also known.

Before we proceed, let us briefly discuss the simplest spherically symmetric model, the spherically symmetric Einstein--null dust system, whose study predates that of the spherically symmetric Einstein--Maxwell--(real) scalar field system. For this system, in the presence of incoming null dust, Hiscock \cite{Hiscock} showed that the metric remains continuous while Christoffel symbols blow up at the Cauchy horizon. In fact, in this setting, curvature components with respect to a parallelly propagated frame blow up. In a subsequent seminal work, Poisson--Israel \cite{PI1,PI2} showed that when another, outgoing, null dust is added and is allowed to interact with the first null dust, generically the Hawking mass is infinite at the Cauchy horizon. This was known as \emph{mass inflation}.

The spherically symmetric Einstein--null dust system, though already gives some insights into the stability and instability properties of Cauchy horizons, is not fully satisfactory even as a spherically symmetric model problem since it does not capture the wave nature of the Einstein equations. A more realistic model is the spherically symmetric Einstein--Maxwell--(real) scalar field system:
\begin{equation}\label{EMSFS}
\begin{cases}
Ric_{\mu\nu}-\f12 g_{\mu\nu} R=2(T^{(sf)}_{\mu\nu}+T^{(em)}_{\mu\nu}),\\
T^{(sf)}_{\mu\nu}=\rd_\mu\phi\rd_\nu\phi-\f 12 g_{\mu\nu} (g^{-1})^{\alp\beta}\rd_\alp\phi\rd_{\beta}\phi,\\
T^{(em)}_{\mu\nu}=(g^{-1})^{\alp\bt}F_{\mu\alp}F_{\nu\bt}-\f 14 g_{\mu\nu}(g^{-1})^{\alp\bt}(g^{-1})^{\gamma\sigma}F_{\alp\gamma}F_{\bt\sigma},\\
\Box_g\phi=0,\,\quad dF=0,\,\quad (g^{-1})^{\alpha\mu}\nabla_\alpha F_{\mu\nu}=0.
\end{cases}
\end{equation}
Here, $\Box_g$ and $\nabla$ respectively denote the Laplace--Beltrami operator and the Levi--Civita connection associated to the metric $g$.

The study of the stability and instability properties of Cauchy horizons in the context of \eqref{EMSFS} in spherical symmetry was initiated in the seminal works of Dafermos \cite{D1,D2}. Taken together with \cite{DRPL}, the work \cite{D2} implies the following theorem:
\begin{theorem}[Dafermos \cite{D2}, Dafermos--Rodnianski \cite{DRPL}]
Given any $2$-ended asymptotically flat future-admissible spherically symmetric initial data to \eqref{EMSFS}, as long as the Maxwell field does not identically vanish, the maximal globally hyperbolic future development has a Cauchy horizon across which the metric is continuously extendible.
\end{theorem}

Here, by \emph{asymptotically flat future-admissible}, we mean that the initial data obey adequate regularity and decay conditions, as well as a global condition (called future admissibility) that ensures, in particular, the existence of a single black hole region in the maximal globally hyperbolic future development like in the case of the Reissner--Nordstr\"om spacetime; see \cite[Definition~3.1]{LO.interior} for the precise definition.

In \cite{LO.interior, LO.exterior}, the first two authors proved the $C^2$ formulation of the strong cosmic censorship conjecture for the Einstein--Maxwell--(real)--scalar field system. Namely, it was proven that generic future-admissible two-ended asymptotically flat initial data lead to a maximal globally hyperbolic development which is $C^2$-future inextendible. 

\begin{theorem}[Luk--Oh \cite{LO.interior, LO.exterior}]\label{thm:SCC}
There exists a generic (in the sense of open and dense in, say, weighted $C^1$ topology; see \cite{LO.interior} for further refinements) class $\mathcal G$ of $2$-ended asymptotically flat future-admissible spherically symmetric initial data such that the maximal globally hyperbolic future development to any initial data in $\mathcal G$ is $C^2$-future-inextendible.
\end{theorem}

In recent breakthrough of Sbierski \cite{Sbierski.C1}, he showed that after further restricting the generic class of data in Theorem~\ref{thm:SCC} to those which are small perturbations of Reissner--Nordstr\"om data, the corresponding solutions are moreover $C^1$-future-inextendible. This thus resolves a $C^1$ formulation of the strong cosmic censorship conjecture in a perturbative (in additional to spherically symmetric) setting.

\subsubsection{The mass inflation problem}\label{sec:mass.inflation} Even though inextendibility properties are the cleanest way to ``measure the strength'' of singularities, it is also of interest (see discussions of the works of Hiscock and Poisson--Israel in the beginning of Section~\ref{sec:SSModel}) to ask whether the Hawking mass blows up at the Cauchy horizon, i.e.,~whether \emph{mass inflation} occurs, for solutions arising from generic data. In particular, this was left open in the works \cite{LO.interior,LO.exterior,Sbierski.C1} discussed above.

At the moment, the best result concerning mass inflation is the following \underline{conditional} result in the seminal work \cite{D2} of Dafermos. It states that mass inflation does occur if a pointwise polynomial lower bound holds along the event horizon:
\begin{theorem}[Dafermos \cite{D2}]\label{thm:intro.dafermos}
Given an asymptotically flat future-admissible initial data set with non-zero charge, if the scalar field satisfies following pointwise lower bound (with respect to an Eddington--Finkelstein-like $v$ coordinate\footnote{For instance, one can choose the $v$ coordinate used in Theorem~\ref{main.theorem.C0.stability}.})
\begin{equation}\label{eq:Dafermos.condition}
|\rd_v\phi\restriction_{\EH_1\cap \{v\geq 1\}}|(v) \geq c v^{-p}
\end{equation}
for some $c>0$ and $p < 9$, then the Hawking mass is identically infinite on the component of the Cauchy horizon $\CH_1$ in Figure~\ref{fig:main.structure} on page~\pageref{fig:main.structure}.
\end{theorem}

If the condition \eqref{eq:Dafermos.condition} holds for generic solutions, then Theorem~\ref{thm:intro.dafermos} would give a generic mass inflation result. In fact, if \eqref{eq:Dafermos.condition} is verified generically, then the blow-up of the Hawking mass can be used to give an alternative proof of the $C^2$ formulation of the strong cosmic censorship conjecture, since the Hawking mass bounds the Kretschmann scalar from below. The works \cite{LO.interior, LO.exterior} however did not follow this path, but instead established a weaker analogue of \eqref{eq:Dafermos.condition}, which gives just an $L^2$-averaged lower bound. This weaker lower bound was slightly easier to prove, and was sufficient for $C^2$-inextendibility, but by itself fell short of establishing generic mass inflation.


To further elaborate the issue, it was proven in \cite{LO.interior, LO.exterior} that for a generic class of initial data, in the corresponding maximal globally hyperbolic future development, the (transversal to Cauchy horizon) derivatives of the scalar field blow up at the Cauchy horizon. However, in principle the scalar field could vanish identically on a portion of the Cauchy horizon near timelike infinity, giving rise to a static Cauchy horizon. In this special scenario, while the Cauchy horizon is still a weak null singularity, the Hawking mass could remain finite (cf.~Proposition~\ref{basic.crit}). Notice that this scenario for which the the first derivatives of the scalar field blow up but the mass remains finite is exactly an analogue of the Hiscock picture for the Einstein--null dust system.

%
%
%

Using Theorem~\ref{thm:main.intro}, we prove the following \underline{conditional} theorem for mass inflation:

\begin{theorem}\label{thm:intro.MI.2}
Consider a $2$-ended asymptotically flat future-admissible spherically symmetric initial data in the generic class $\mathcal G$ in Theorem~\ref{thm:SCC} so that the scalar field and its derivatives are initially compactly supported. \underline{Assume} that higher derivatives of the scalar field exhibit ``improved decay'' along each connected component of the event horizon (see precise assumptions in Theorem~\ref{mass.infl.thm}).

Then the Hawking mass is identically infinite at each connected component of the Cauchy horizon.
\end{theorem}

Theorem~\ref{thm:intro.MI.2} is proven using part (2) of Theorem~\ref{thm:main.intro}. They are related because Theorem~\ref{thm:main.intro} rules out the possibility --- for the linear scalar wave equation \eqref{wave.eqn} --- that the scalar field vanishes identically on a portion of the Cauchy horizon near timelike infinity. We show using a perturbative argument that this also holds for the \emph{nonlinear} solution. Therefore by our previous discussion mass inflation must occur.

The improved decay assumption that we need in Theorem~\ref{thm:intro.MI.2} corresponds to the upper bound assumptions of Theorem~\ref{thm:main.intro} on the event horizon. It should be noted that while the lower bound proven in \cite{LO.interior} is expected to be sharp (for instance by comparing with the linear result in \cite{AAG}; see Remark~\ref{rmk:expectation.AAG}), the upper bounds proven in \cite{DRPL} are worse than those expected. The assumptions of Theorem~\ref{thm:intro.MI.2} require some improvements over the upper bounds in \cite{DRPL} (although it still does not require the sharp upper bounds). Notice that one only needs an improved \emph{upper bound}, and in principle that is easier to obtain as compared to the improved \emph{lower bound} required in Theorem~\ref{thm:intro.dafermos}.

It should be stressed, however, that main point of Theorem~\ref{thm:intro.MI.2} is not mass inflation per se, since the assumed conditions in Theorems~\ref{thm:intro.dafermos} and \ref{thm:intro.MI.2} are both expected to hold; see Remark~\ref{rmk:expectation.AAG} below. Instead, we want to demonstrate with this theorem how linear result of the type in Theorem~\ref{thm:main.intro} can be applied in a nonlinear setting quite easily.

We end this subsection with a few remarks on Theorem~\ref{thm:intro.MI.2}.
\begin{remark}\label{rmk:expectation.AAG}
Both the assumed pointwise lower bound in Theorem~\ref{thm:intro.dafermos} and the assumed improved decay in Theorem~\ref{thm:intro.MI.2} remain open problems. Nevertheless, the results in \cite{AAG} for the linear wave equation on Reissner--Nordstr\"om suggest that the following may be true:
\begin{itemize}
\item Solutions arising from generic data obey the pointwise lower bound in Theorem~\ref{thm:intro.dafermos}.
\item Improved decay estimates assumed in Theorem~\ref{thm:intro.MI.2} hold for \underline{all} initial data (not just generic data).
\end{itemize} 
\end{remark}

\begin{remark}
Note that Theorems~\ref{thm:intro.dafermos} and \ref{thm:intro.MI.2} have no analogue for the Einstein--null dust system: indeed while a nust dust is a good approximation to a scalar field in the high frequency limit, Theorems~\ref{thm:intro.dafermos} and \ref{thm:intro.MI.2} precisely capture a phenomenon regarding the behavior of the scalar field at zero frequency.
\end{remark}

\begin{remark}
Note that the global structure of the interior of dynamical black hole in question may be very different from the global structure of Reissner--Nordstr\"om. (In particular, unlike Reissner--Nordstr\"om, its boundary can in principle have a spacelike portion!) Although our Theorem~\ref{thm:main.intro} is proved using global considerations in the Reissner--Nordstr\"om interior, it is still applicable to the problem at hand because in the course of the proof of Theorem~\ref{thm:SCC}, it is also established that \emph{in a region sufficiently close to timelike infinity}, the spacetime metric is indeed a small perturbation --- in some rough norms --- of that of Reissner--Nordstr\"om. Due to the monotonicity of the Hawking mass, in order to establish mass inflation, it suffices to consider the region near timelike infinity, in which we can use a perturbative argument.
\end{remark}

\begin{remark}
While we have only studied here the system \eqref{EMSFS}, it is also of interest to go beyond it and study the Einstein--Maxwell--\emph{charged} scalar field system in spherical symmetry, including in the case when the scalar field is massive. The more general system allows one to study simultaneously gravitational collapse and strong cosmic censorship. The issues regarding strong cosmic censorship conjecture for this system has recently been studied in a series of papers of Van de Moortel \cite{VDM,VDM2,VDM3,VDM4}, and the paper of Kehle--Van de Moortel \cite{KehleVdM}, which in particular show both $C^0$ stability and $C^2$ inextendibility conditional on appropriate decay assumptions on the event horizon. Spectacularly, it was shown in \cite{VDM3} that for asymptotically Reissner--Nordstr\"om black holes arising from one-ended gravitational collapse, the weak null singularities along the Cauchy horizon must break down. We remark that despite all this important progress, for the Einstein--Maxwell--charged scalar field system in spherical symmetry, it remains an open problem whether the Hawking mass generically blows up identically at the Cauchy horizon.
\end{remark}

\subsection{Outline of the paper} The remainder of the paper is organized as follows. First, in \textbf{Section~\ref{sec.geometry}}, we will introduce the geometry of the interior of Reissner--Nordstr\"om. In \textbf{Section~\ref{sec:scattering}}, we establish some simple bounds with energy estimates and use them to define the scattering maps in the interior of Reissner--Nordstr\"om for spherically symmetric data. In \textbf{Section~\ref{sec:thm}}, we will state a precise version of Theorem~\ref{thm:main.intro} (Theorem~\ref{thm:main}). In \textbf{Section~\ref{sec:phase.space}}, we discuss the scattering map for spherically symmetric data in phase space. Using this, we prove Theorem~\ref{thm:main} in \textbf{Section~\ref{sec:proof}}. Finally, we apply our Theorem~\ref{thm:main} to the spherically symmetric Einstein--Maxwell--(real) scalar field system to obtain a conditional mass inflation result in \textbf{Section~\ref{sec.application}}. 

\subsection*{Acknowledgements} We thank Mihalis Dafermos and Jan Sbierski for many insightful discussions.

J.~Luk is supported by a Terman fellowship and the NSF Grant DMS-2005435. S.-J.~Oh is supported by the Samsung Science and Technology Foundation under Project Number SSTF-BA1702-02, a Sloan Research Fellowship and a NSF CAREER Grant DMS-1945615. Y.~Shlapentokh-Rothman acknowledges support from NSF grant DMS-1900288, from an Alfred P. Sloan Fellowship in Mathematics, and from NSERC discovery grants RGPIN-2021-02562 and DGECR-2021-00093. 

\section{The geometry of the interior of Reissner--Nordstr\"om}\label{sec.geometry}

Let $M$ and ${\bf e}$ be real numbers satisfying $0<|{\bf e}|<M$. We define the \emph{interior of Reissner--Nordstr\"om with parameter $M$ and ${\bf e}$} to be the Lorentizian manifold $(\mathcal M_{RN}, g_{RN})$, where 
\begin{itemize}
\item $\mathcal M_{RN} \doteq \mathbb R\times (r_-,r_+)\times \mathbb S^2$, where $r_\pm \doteq M\pm\sqrt{M^2-{\bf e}^2}$; 
\item the metric $g_{RN}$ of the Reissner--Nordstr\"om spacetime given by:
\begin{equation}\label{RN.metric}
g_{RN} \doteq -\left(1-\f {2M} r+\f{{\bf e}^2}{r^2}\right)\,dt\otimes dt + \left(1-\f {2M} r+\f{{\bf e}^2}{r^2}\right)^{-1}\, dr\otimes dr + r^2 g_{\mathbb S^2},
\end{equation}
where $t\in \mathbb R$, $r\in (r_-,r_+)$ and $g_{\mathbb S^2}$ is the metric on the standard round sphere of radius $1$.
\end{itemize}
Together with an appropriate Maxwell field, Reissner--Nordstr\"om is a solution to the Einstein--Maxwell system.

\subsection{Spherical symmetry and the quotient manifold}

The Reissner--Nordstr\"om spacetime $(\mathcal M_{RN},g_{RN})$ is easily seen to be spherically symmetric\footnote{See Section~\ref{def.SS} for a general discussion.} in the sense that for $\mathcal M_{RN} = \mathcal Q_{RN} \times \mathbb S^2$ and $\mathcal Q_{RN}=\mathbb R\times (r_-,r_+)$, we can write
$$g_{RN} = g_{\mathcal Q_{RN}} + r^2 g_{\mathbb S^2},$$
where 
\begin{itemize}
\item $(\mathcal Q_{RN},g_{\mathcal Q_{RN}})$ is a $(1+1)$-dimensional Lorentzian manifold with $g_{\mathcal Q_{RN}}$ given by 
$$g_{\mathcal Q_{RN}} = -\left(1-\f {2M} r+\f{{\bf e}^2}{r^2}\right)\,dt\otimes dt + \left(1-\f {2M} r+\f{{\bf e}^2}{r^2}\right)^{-1}\, dr\otimes dr;$$
\item given a point $p\in \mathcal M_{RN}$, $r(p)$ depends only on ${\boldsymbol \pi}(p)$, where ${\boldsymbol \pi}:\mathcal M_{RN}\to \mathcal Q_{RN}$ is the natural projection map.
\end{itemize}
We will denote by $\vartheta$ a point on $\mathbb S^2$. Frequently, we will also use the standard spherical coordinates $(\theta,\varphi)$, in which case we have
$$g_{\mathbb S^2} = d\tht^2 + \sin^2\tht \, d\varphi^2.$$

\subsection{The $(u,v)$ coordinate system}\label{sec.null.1}

We define the $r^*$ coordinate in the interior of the Reissner--Nordstr\"om black hole:
\begin{equation}\label{eq:r*.RN.def}
r^* \doteq r+(M+\frac{2M^2-{\bf e}^2}{2\sqrt{M^2-{\bf e}^2}})\log (r_+-r) +(M-\frac{2M^2-{\bf e}^2}{2\sqrt{M^2-{\bf e}^2}})\log (r-r_-).
\end{equation}
Notice that this implies
$$\f{dr}{dr^*} = \f{r^2-2Mr+{\bf e}^2}{r^2}.$$

Define then the null coordinates
\begin{equation}\label{eq:uv.RN.def}
v \doteq \frac 12(r^*+t),\quad u \doteq \frac 12(r^*-t),
\end{equation}
which implies
\begin{equation}\label{null.vector.fields}
\frac{\rd}{\rd v}= \frac{\rd}{\rd r^*}+\frac{\rd}{\rd t},\quad\frac{\rd}{\rd u}= \frac{\rd}{\rd r^*}-\frac{\rd}{\rd t}.
\end{equation}
According to \eqref{RN.metric}, in this coordinate system, the Riessner--Nordstr\"{o}m metric takes the form
$$g_{RN}=-\f{\Omg_{RN}^2}{2} (du\otimes dv+dv\otimes du)+ r_{RN}^2 g_{\mathbb S^2},$$
where $r_{RN} = r$ is now thought of as a function of $(u,v)$, and $\Omg_{RN}^2 \doteq -4(1-\frac{2M}{r_{RN}}+\frac{{\bf e}^2}{r_{RN}^2})$. Moreover, by \eqref{eq:r*.RN.def} and \eqref{null.vector.fields}, we have
\begin{equation}\label{int.r}
\rd_v r_{RN} = \rd_u r_{RN} = 1-\f{2M}{r_{RN}}+\f{{\bf e}^2}{r_{RN}^2}.
\end{equation}

In the $(u,v)$ coordinates, the spherically symmetric wave equation $\Box_{g_{RN}} \phi = 0$ takes the form
\begin{equation}\label{eq:ss.wave}
\rd_u \rd_v \phi + \f{\rd_v r_{RN}}r \rd_u \phi + \f{\rd_u r_{RN}}r \rd_v \phi = 0.
\end{equation}

\subsection{Event horizon and Cauchy horizon}\label{sec:horizons} We attach boundaries to $\mathcal Q_{RN}$, known as \emph{event horizon} $\EH_{total}$ and \emph{Cauchy horizon} $\CH_{total}$, to obtain a manifold-with-corner $\overline{\mathcal Q_{RN}}$. Define $\overline{\mathcal M_{RN}} \doteq \overline{\mathcal Q_{RN}}\times \mathbb S^2$. Abusing conventions slightly, we will also refer to $\EH_{total} \times \mathbb S^2 \subset \overline{\mathcal M_{RN}}$ as the event horizon, and $\CH_{total} \times \mathbb S^2 \subset \overline{\mathcal M_{RN}}$ as the Cauchy horizon. Notice that the metric $g_{RN}$ extends smoothly up to the boundary.

Define the functions $U_{\EH}(u)$, $U_{\CH}(u)$, $V_{\EH}(v)$ and $V_{\CH}(v)$ which are smooth and strictly increasing functions of their arguments and satisfy the following ODEs: 
\begin{equation}\label{U.EH.def}
\f{dU_{\EH}}{du}=e^{2\kappa_+ u} \mbox{ and }\, U_{\EH}(u)\to 0 \mbox{ as }u\to -\infty;
\end{equation}
\begin{equation}\label{U.CH.def}
\f{dU_{\CH}}{du}=e^{-2\kappa_- u} \mbox{ and }\, U_{\CH}(u)\to 1 \mbox{ as }u\to +\infty;
\end{equation}
\begin{equation}\label{V.EH.def}
\f{dV_{\EH}}{dv}=e^{2\kappa_+ v}\mbox{ and }V_{\EH}(v)\to 0\mbox{ as }v\to -\infty;
\end{equation}
\begin{equation}\label{V.CH.def}
\f{dV_{\CH}}{dv}=e^{-2\kappa_- v}\mbox{ and }V_{\CH}(v)\to 1\mbox{ as }v\to +\infty,
\end{equation}
where $\kappa_+>0$ and $\kappa_->0$ are defined to be to be\footnote{Note that this coincides with the definition in \cite{LO.interior}, but differs from that in \cite{LO.instab}, where $\kappa_-$ is taken to be negative.} 
\begin{equation}\label{eq:def.kappa}
\kappa_+ \doteq \f{r_+-r_-}{2r_+^2},\quad \kappa_- \doteq \f{r_+-r_-}{2r_-^2}.
\end{equation}

In the $(U_{\EH},V_{\EH})$ coordinate system, we attach the boundaries $\EH_1\doteq \{U_{\EH}=0\}$ and $\EH_2\doteq \{V_{\EH}=0\}$. Denote also the \emph{event horizon} as $\EH_{total}=\EH_1\cup\EH_2$. In the $(U_{\CH},V_{\CH})$ coordinate system, we attach the boundaries $\CH_1 \doteq \{V_{\CH}=1\}$ and $\CH_2  \doteq \{U_{\CH}=1\}$. Denote also the \emph{Cauchy horizon} as $\CH_{total}=\CH_1\cup\CH_2$. 

We define the \emph{bifurcation sphere of $\EH_{total}$} by
$$\mathcal B_{\EH} \doteq \EH_1\cap \EH_2 = \{(U_{\EH},V_{\EH}):U_{\EH}=V_{\EH}=0\}.$$
Note that $\mathcal B_{\EH}$ is a subset of both $\EH_1$ and $\EH_2$. Similarly, define \emph{bifurcation sphere of $\CH_{total}$} by
$$\mathcal B_{\CH} \doteq \CH_1\cap \CH_2 = \{(U_{\CH},V_{\CH}):U_{\CH}=V_{\CH}=1\}.$$ 

\subsection{Behavior of $\Omg_{RN}$ near the horizons}
Using \eqref{eq:r*.RN.def} and recalling \eqref{eq:def.kappa}, it is easy to check that $r^*$ can then be alternatively written as
\begin{equation}\label{eq:def.r^*}
r^*=r_{RN}+\f{1}{2\kappa_+}\log (r_+-r_{RN}) -\f{1}{2\kappa_-}\log (r_{RN}-r_-).
\end{equation}
We compute that as $r_{RN}\to r_+$, we have
$$r_+-r_{RN}=e^{-2\kappa_+ r_+}(r_+-r_-)^{\f{\kappa_+}{\kappa_-}}e^{2\kappa_+ r^*}(1+O(r_+-r_{RN})).$$
In other words, for any $A\in \mathbb R$, in the $r^* \leq A$ region (i.e.,~in a neighborhood of the event horizon),
\begin{equation}\label{eq:Omg.near.H}
\f 14\Omg_{RN}^2=-\rd_u r_{RN}=-\rd_v r_{RN}=\f{e^{-2\kappa_+ r_+}(r_+-r_-)^{1+\f{\kappa_+}{\kappa_-}}}{r_+^2}e^{2\kappa_+ r^*}(1+O_A(r_+-r_{RN}))= O_A(e^{2\kappa_+ r^*}).
\end{equation}

On the other hand, as $r_{RN}\to r_-$, we have
$$r_{RN}-r_-=e^{2\kappa_-r_-}(r_+-r_-)^{\f{\kappa_-}{\kappa_+}}e^{-2\kappa_-r^*}(1+O(r_{RN}-r_-)).$$
As a consequence, for any $A\in \mathbb R$, in the $r^* \geq A$ region (i.e.,~in a neighborhood of the Cauchy horizon),
\begin{equation}\label{eq:Omg.near.CH}
\f 14\Omg_{RN}^2=-\rd_u r_{RN}=-\rd_v r_{RN}=\f{e^{-2\kappa_-r_-}(r_+-r_-)^{1+\f{\kappa_-}{\kappa_+}}}{r_-^2}e^{-2\kappa_-r^*}(1+O_A(r_{RN}-r_-))= O_A(e^{-2\kappa_-r^*}).
\end{equation}

\subsection{Killing vector fields on $(\mathcal M_{RN},g_{RN})$}\label{sec:vectorfields}

The Reissner--Nordstr\"om interior $(\mathcal M_{RN}, g_{RN})$ admits the following Killing vector fields:
\begin{equation}\label{eq:def.Killing}
T \doteq \rd_t,\quad \mathcal O_1 \doteq \sin\varphi \rd_{\theta} + \f{\cos\varphi\cos\theta}{\sin\theta}\rd_{\varphi},\quad \mathcal O_2 \doteq \cos\varphi\rd_{\theta}-\f{\sin\varphi\cos\theta}{\sin\theta}\rd_{\varphi},\quad \mathcal O_3 \doteq \rd_{\varphi}.
\end{equation}

Note that all of these vector fields extend smoothly up to the event horizon and the Cauchy horizon. In particular\footnote{Here, this is to be understood as the extension of the coordinate vector field $\rd_v$ in the $(u,v)$ coordinate system to $\overline{\mathcal M_{RN}}$. Similar convention is used for $\rd_u$ on $\EH_2$ and $\CH_1$.}, 
\begin{equation}\label{T.on.the.horizons}
\mbox{$T = \f 12\rd_v$ on\ $\EH_1$ and $\CH_2$, while $T = - \f 12\rd_u$ on $\EH_2$ and $\CH_1$.}
\end{equation}

\subsection{Volume forms}\label{sec:vol}

Before we end this section on the Reissner--Nordstr\"om geometry, we briefly comment on volume forms on $(\mathcal M_{RN},g_{RN})$.

The metric $g_{RN}$ induces a natural (positive) volume form $\vol$. In the $(u,v,\theta,\varphi)$ coordinates, we have
\begin{equation}\label{vol.est}
\vol \doteq \sqrt{-\det g_{RN}} \, du\,dv\, d\tht\, d\varphi \, = -2(1-\f{2M}{r_{RN}}+ \f{{\bf e}^2}{r_{RN}^2}) r_{RN}^2 \sin\tht \, du\,dv\, d\tht\, d\varphi\,= \, \f 12 r_{RN}^2\Omg_{RN}^2 \sin\tht\, du\,dv\, d\tht\, d\varphi.
\end{equation}
On constant-$u$ and constant-$v$ null hypersurfaces, we define positive volume forms $\vol_u$ and $\vol_v$ respectively by $\vol = du\wedge \vol_u$ and $\vol = dv\wedge \vol_v$. It can then be checked that
\begin{equation}\label{eq:vol.u.v}
\begin{split}
\vol_u \doteq -2(1-\f{2M}{r_{RN}}+ \f{{\bf e}^2}{r_{RN}^2}) r_{RN}^2 \sin\tht \,dv\, d\tht\, d\varphi = \f 12 r_{RN}^2\Omg_{RN}^2 \sin \tht \,dv\, d\tht\, d\varphi,\\
\vol_v \doteq -2(1-\f{2M}{r_{RN}}+ \f{{\bf e}^2}{r_{RN}^2}) r_{RN}^2 \sin\tht \,du\, d\tht\, d\varphi = \f 12 r_{RN}^2\Omg_{RN}^2 \sin \tht\, du\, d\tht\, d\varphi.
\end{split}
\end{equation}
We will also use $d\sigma_{\vartheta}$ to denote the standard volume form of the induced metric on the sphere of symmetry. Note that
$$d\sigma_{\vartheta} \doteq r_{RN}^2 \sin\tht \, d\tht\, d\varphi.$$

\section{Energy estimates for spherically symmetric solutions in the interior of Reissner--Nordstr\"om}\label{sec:scattering}

In this section, we discuss a scattering theorem for the \eqref{wave.eqn} on $(\mathcal M_{RN},g_{RN})$ for spherically symmetric solutions. We first review the vector field multiplier method (specialized to the spherically symmetric case) in \textbf{Section~\ref{sec:vector.field.method}}. In \textbf{Section~\ref{sec:SS.EE}}, we derive energy estimates for solutions to the wave equation with spheically symmetric data. Finally, in \textbf{Section~\ref{sec:scat.ss}}, we define the transmission and reflection maps.

\subsection{The vector field multiplier method}\label{sec:vector.field.method}

In this subsection, we review some general notions regarding the vector field multiplier method. These will be useful not only for Theorem~\ref{scat.SS}, but also for the rest of the paper.

We begin with the definition of stress--energy--momentum tensor on a general Lorentzian manifold:
\begin{definition}[Stress--energy--momentum tensor]
Define the \emph{stress--energy--momentum tensor} by
$$\mathbb T_{\mu\nu}[\phi] \doteq \rd_\mu\phi\rd_\nu\phi-\f 12 g_{\mu\nu} (g^{-1})^{\alp\bt}\rd_{\bt}\phi\rd_{\alp}\phi.$$
\end{definition}
The following easy standard positivity property will be useful:
\begin{lemma}\label{positivity}
Let $X$, $Y$ be future-directed and causal at a point $p$, then $(\T[\phi](X,Y))(p)\geq 0$.
\end{lemma}
Denoting by $\nabla$ the Levi-Civita connection associated to $g$, we have
$$\nabla^\mu \mathbb T_{\mu\nu}[\phi]=\Box_g\phi\rd_\nu\phi,$$ 
which implies the following \emph{energy identity}:
\begin{lemma}[Energy estimates]\label{lem:energy.identity}
For a compact region $\mathcal D \subseteq \mathcal M$ with piecewise smooth boundary $\partial \mathcal D$, which is oriented with respect to the outward pointing normal, Stokes' theorem now yields
\begin{equation}
\label{EnergyEst}
\int\limits_{\partial D} \iota_{\T[\phi](X, \cdot)^\sharp} \,\vol = \int\limits_{\mathcal D} d\big(  \iota_{\T[\phi](X, \cdot)^\sharp} \,\vol\big) = \int\limits_{\mathcal D} \Big(\T[\phi]_{\mu \nu} \nabla^\mu X^\nu + \Box_g \phi (X\phi) \Big) \, \vol \;,
\end{equation}
where $\vol$ is the volume form induced by the metric $g$.
\end{lemma}

In order to apply the energy estimates \eqref{EnergyEst}, it is convenient to make the following definition:
\begin{definition}
Given a $C^1$ vector field $X$, define the \emph{deformation tensor} ${ }^{(X)}\pi$
$${ }^{(X)}\pi_{\mu\nu} \doteq \nabla_\mu X_{\nu} +\nabla_\nu X_\mu.$$
\end{definition}

On $(\mathcal M_{RN},g_{RN})$, ${ }^{(X)}\pi_{\mu\nu}$ can be explicitly computed as follows (see \cite{Fra}):
\begin{lemma}\label{RN.K}
For every spherically symmetric $C^1$ vector field $X$ on $(\mathcal M_{RN},g_{RN})$, the following identity holds on $(\mathcal M_{RN},g_{RN})$:
\begin{equation*}
\begin{split}
\mathbb T_{\mu\nu}[\phi]{ }^{(X)}\pi^{\mu\nu}=&-\f{4}{\Omg_{RN}^2}\left((\rd_uX^v)(\rd_v\phi)^2+(\rd_vX^u)(\rd_u\phi)^2\right)-\f {4}{r_{RN}} (X^v+X^u)(\rd_u\phi\rd_v\phi)\\
&- 2\left(\f 12(\rd_vX^v+\rd_uX^u)+X^v \rd_v\log\Omg_{RN} + X^u \rd_u\log\Omg_{RN} \right)|\slashed\nabla\phi|^2,
\end{split}
\end{equation*}
where $|\slashed\nabla\phi|^2$ is the square of the norm of the angular derivatives with respect to the induced metric on the $2$-spheres of symmetry, and is given in the $(\theta,\varphi)$ coordinates by
$$|\slashed\nabla\phi|^2 \doteq \f{1}{r_{RN}^2}(\rd_\theta\phi)^2 + \f{1}{r_{RN}^2\sin^2\theta}(\rd_\varphi\phi)^2.$$
\end{lemma}

\subsection{Energy estimates}\label{sec:SS.EE}

In this section, we use the formalism defined in Section~\ref{sec:vector.field.method} above to prove energy estimates for spherically symmetric solutions to the wave equation $\Box_{g_{RN}} \phi = 0$.

We remark that the energy estimates presented in Theorem~\ref{scat.SS} are by now standard; see for instance \cite{Fra, KSR}. One particular feature of the energy estimates we use is that we apply a multipler vector field which is \emph{non-smooth} ($C^0$ but not $C^1$) at the horizons (see the definitions \eqref{y.def} and \eqref{Y.def}). This type of non-smooth multiplier, which appeared already in \cite[Proposition~9.2]{LO.interior} (see also \cite{DRS}), generates some better spacetime terms, which will be useful later in Section~\ref{sec:EE.inhom}.

Before we proceed, let us define some useful weight functions:
\begin{definition}[Polynomial weights]\label{def:w.weight}
Given $p \in [0,\infty)$, let $w_p:\mathbb R\to \mathbb R$ be a smooth and non-decreasing function such that
$$w_p(x)=\begin{cases}
1 \quad &\mbox{ if }x \leq \f 12\\
2 x^{p} \quad &\mbox{ if }x\geq 1.
\end{cases}
$$

\end{definition}

We are now ready to give the main energy estimates.

\begin{theorem}[Energy estimates in spherical symmetry]\label{scat.SS}
Fix $p_1,\,p_2\in (1,+\infty)$. Let $\phi$ be a solution to 
$$\Box_{g_{RN}} \phi = 0$$
in $\mathcal M_{RN}$ which is smooth up to $\EH_{total} = \EH_1 \cup \EH_2$ with spherically symmetric characteristic initial data on $\EH_{total}$ such that the following holds:
\begin{enumerate}
\item $$\phi\restriction_{\mathcal B_{\EH}} = 0,\quad \lim_{v\to +\infty} \phi\restriction_{\EH_1}(v) = 0 = \lim_{u\to +\infty} \phi\restriction_{\EH_2}(u),$$
\item \begin{equation}\label{eq:def.A.in.EE}
\int_{\mathbb R} w_{p_1}(v)w_{p_2}(-v)(T\phi\restriction_{\EH_1})^2(v)\,dv + \int_{\mathbb R} w_{p_1}(-u)w_{p_2}(u)(T\phi \restriction_{\EH_2})^2(u)\,du \doteq A <+\infty.
\end{equation}
\end{enumerate}

Then the following holds for some $C>0$ depending only on $p_1$, $p_2$, $M$ and $\bf e$:
\begin{enumerate}
\item The following uniform upper bound holds:
\begin{equation}\label{eq:main.energy.est}
\sup_{u\in \mathbb R} \int_{\mathbb R} w_{p_1}(v)w_{p_2}(-v)(\rd_v\phi)^2(u,v)\,dv + \sup_{v\in \mathbb R} \int_{\mathbb R} w_{p_1}(-u)w_{p_2}(u)(\rd_u\phi)^2(u,v)\,d u \leq CA.
\end{equation}
\item The solution $\phi$ extends to a continuous function (which we abuse notation slightly and write) $\phi: \mathcal M_{RN}\cup \EH_{total} \cup \CH_{total}\to \mathbb R$.
\item The derivative $\rd_v\phi$ extends continuously to $\CH_2\setminus \CH_1$ and the derivative $\rd_u \phi$ extends continuously to $\CH_1 \setminus \CH_2$.
\item The extension of $\phi$ satisfies $\phi\restriction_{\CH_1},\,\phi \restriction_{\CH_2}\in W^{1,2}_{loc}$ and obeys the estimate 
\begin{equation}\label{eq:energy.est.on.CH}
\int_{\mathbb R} w_{p_1}(v)w_{p_2}(-v)(T\phi\restriction_{\CH_2})^2(v)\,dv + \int_{\mathbb R} w_{p_1}(-u)w_{p_2}(u)(T\phi \restriction_{\CH_1})^2(u)\,du \leq CA.
\end{equation}
\end{enumerate}
\end{theorem}
\begin{proof}
In this proof, constants $C$ and implicit constants in $\ls$ depend only on $p_1$, $p_2$, $M$ and $\bf e$.

\pfstep{Step~1: Energy estimates and proof of \eqref{eq:main.energy.est}} Fix $\sigma>1$ and define the function $y(r^*)$ by
\begin{equation}\label{y.def}
y(r^*) \doteq 2-\f{\sigma-1}{2}\int_{-\infty}^{r^*} (1+|(r^*)'|)^{-\sigma}\, d(r^*)'.
\end{equation}
Here, $r^*$ is a function of $u$, $v$ given by \eqref{eq:uv.RN.def}. Notice that \eqref{y.def} is well-defined since $\sigma>1$. Moreover, 
\begin{enumerate}
\item $y$ is decreasing in $r^*$, 
\item $1\leq y\leq 2$, and
\item $y\restriction_{\EH} = 2$ and $y\restriction_{\CH} = 1$.
\end{enumerate}

Define the vector field $Y$ by 
\begin{equation}\label{Y.def}
Y \doteq y^N(r^*) \left(w_{p_1}(v)w_{p_2}(-v)\rd_v + w_{p_1}(-u)w_{p_2}(u)\rd_u\right).
\end{equation}
First note that $Y^v = y^N(r^*) w_{p_1}(v)w_{p_2}(-v)\geq 0$ and $Y^u = y^N(r^*) w_{p_1}(-u)w_{p_2}(u) \geq 0$, i.e.,~$Y$ is future-directed and causal. Hence, by Lemma~\ref{positivity}, we get a non-negative energy. Indeed, by Lemma~\ref{lem:energy.identity} with $X = Y$, and using \eqref{vol.est}, \eqref{eq:vol.u.v}, we obtain that there is some $c_0>0$ such that for any $u_*,\,v_*\in \mathbb R$,
\begin{equation}\label{eq;main.EE.before.est.bulk}
\begin{split}
&\: \int_{\{u=u_*\}} y^N w_{p_1}(v)w_{p_2}(-v) r_{RN}^2 (\rd_v\phi)^2(u,v) \, dv + \int_{\{v=v_*\}} y^N w_{p_1}(-u)w_{p_2}(u) r_{RN}^2(\rd_u\phi)^2(u,v) \, du \\
= &\:  \int_{\EH_1} 2^N w_{p_1}(v)w_{p_2}(-v) r_+^2 (\rd_v\phi)^2(u,v) \, dv + \int_{\EH_2} 2^N w_{p_1}(-u)w_{p_2}(u) r_+^2 (\rd_u\phi)^2(u,v) \, du \\
&\: - \frac{1}{2} \int_{\{(u,v): u\in (-\infty,u_*),\,v\in (-\infty,v_*)\}} \mathbb T_{\mu\nu}{ }^{(Y)}\pi^{\mu\nu} \,r_{RN}^2 \Omg_{RN}^2 \, du\, dv \\
=&\: 2^{N+2} r_+^2 A - \frac{1}{2} \int_{\{(u,v): u\in (-\infty,u_*),\,v\in (-\infty,v_*)\}} \mathbb T_{\mu\nu}{ }^{(Y)}\pi^{\mu\nu} \,r_{RN}^2 \Omg_{RN}^2 \, du\, dv,
\end{split}
\end{equation}
where in the penultimate identity we used properties of $y$ above, and in the last line we used \eqref{T.on.the.horizons} and \eqref{eq:def.A.in.EE}.

We claim that there exists $N$ sufficiently large and $c>0$ (both depending on $\sigma$, $M$ and ${\bf e}$) such that
\begin{equation}\label{bulk.term.to.prove}
\mathbb T_{\mu\nu}{ }^{(Y)}\pi^{\mu\nu}r_{RN}^2 \Omg_{RN}^2\geq \f{c}{(1+|r^*|)^{\sigma}}(w_{p_1}(v)w_{p_2}(-v)(\rd_v\phi)^2+w_{p_1}(-u)w_{p_2}(u)(\rd_u\phi)^2)\geq 0.
\end{equation}
The verification of \eqref{bulk.term.to.prove} will be postponed to Step~1(a) below.

Accepting \eqref{bulk.term.to.prove} for the moment and fixing $N$ such that \eqref{bulk.term.to.prove} holds, the desired estimate \eqref{eq:main.energy.est} is an immediate consequence of \eqref{eq;main.EE.before.est.bulk}, $y\geq 1$ and $r\geq r_-$.

\pfstep{Step~1(a): Controlling the bulk term in the energy estimates} Our goal now is to prove \eqref{bulk.term.to.prove}. Since $\phi$ is spherically symmetric, by Lemma \ref{RN.K}, we have
\begin{equation}\label{KY}
\begin{split}
\mathbb T_{\mu\nu}{ }^{(Y)}\pi^{\mu\nu}=&-\f{4}{\Omg_{RN}^2}\left((\rd_uY^v)(\rd_v\phi)^2+(\rd_vY^u)(\rd_u\phi)^2\right)-\f {4}r (Y^v+Y^u)(\rd_u\phi\rd_v\phi)\\
=&\underbrace{\f{2N(\sigma-1)}{\Omg_{RN}^2}y^{-1}(1+|r^*|)^{-\sigma}\left(Y^v(\rd_v\phi)^2+Y^u(\rd_u\phi)^2\right)}_{Main\,term}-\underbrace{\f {2}r (Y^v+Y^u)(\rd_u\phi\rd_v\phi)}_{Error\,term}.
\end{split}
\end{equation}
Note that $r_{RN}^2 \Omg_{RN}^2 \mbox{(Main term in \eqref{KY})}\gtrsim \mbox{(RHS of \eqref{bulk.term.to.prove})}$ since $r_{RN}$ and $y$ are both bounded above and away from $0$. Therefore, in order to prove \eqref{bulk.term.to.prove}, it suffices to show that for $N$ sufficiently large, the ``Error term'' in \eqref{KY} can be dominated by the $\f 12 \times$``Main term''. We consider the following cases, which exhaust all possibilities (although not mutually exclusive):

\underline{Case 1: $v\geq 0$, $v\geq 2|u|$.} In this region, $r$ is bounded away from $r_+$. Therefore, $\Omg_{RN}^2\lesssim e^{-2\kappa_-(v+u)}$ by \eqref{eq:Omg.near.CH} (where here, and below, the implicit constant depends on $M$ and ${\bf e}$, but is independent of $u$ and $v$). Using $v\geq 2|u|$, this implies $\Omg_{RN}^2\lesssim \min\{e^{-\kappa_- v},e^{-2\kappa_- |u|}\}$. Since $r^*=v+u$, there exists $c>0$ such that $\f 1{\Omg^2 y (1+|r^*|)^{\sigma}}\gtrsim \max\{e^{cv},e^{c|u|}\}$. This implies 
\begin{equation}\label{key.lower.bound}
\min\{\f {Y^v}{\Omg_{RN}^2 y (1+|r^*|)^{\sigma}},\, \f {Y^u}{\Omg_{RN}^2 y (1+|r^*|)^{\sigma}} \}\gtrsim\max\{Y^v, \, Y^u\}.
\end{equation}
Therefore, by choosing $N>0$ sufficiently large and using the Cauchy--Schwarz inequality for the ``Error term'', one sees that \eqref{KY} is positive in this region.

\underline{Case 2: $v\leq 0$, $|v|\geq 2|u|$.} In this region, $r$ is bounded away from $r_-$ and hence $\Omg^2\lesssim e^{2\kappa_+(v+u)}$. This then implies $\Omg^2\lesssim \min\{e^{-\kappa_+ |v|},e^{-2\kappa_+|u|}\}$. As a consequence, \eqref{key.lower.bound} holds and the rest of the proof proceeds as in Case 1.

\underline{Case 3: $u\geq 0$, $u\geq 2|v|$.} This can be treated similarly as Case 1. 

\underline{Case 4: $u\leq 0$, $|u|\geq 2|v|$.} This can be treated similarly as Case 2. 

\underline{Case 5: $uv\leq 0$, $|u|\leq 2|v|\leq 4|u|$.} In this region, we have $Y^v\sim Y^u$ and $\f 1{\Omg^2 y (1+|r^*|)^{\sigma}}\gtrsim 1$. As a consequence, \eqref{key.lower.bound} holds and the rest of the proof proceeds as in Case 1.

\underline{Case 6: $uv\geq 0$.} Since $u$ and $v$ have the same sign, one checks that there exists $c>0$ such that $\f 1{\Omg^2 y (1+|r^*|)^{\sigma}}\gtrsim e^{c(|v|+|u|)}$. In particular, \eqref{key.lower.bound} holds and the rest of the proof proceeds as in Case 1.

We have thus verified the claim \eqref{bulk.term.to.prove}.

\pfstep{Step~2: Proof of continuous extendibility of $\phi$} Continuous extendibility is a direct consequence of \eqref{eq:main.energy.est}. Indeed, \eqref{eq:main.energy.est} and the Cauchy--Schwarz inequality imply that for $u,\,u',\,v,\,v'\in [ 1, +\infty)$,
\begin{equation}\label{eq:phi.Cauchy}
|\phi(u,v) - \phi(u',v')| \ls |v^{-\f{p_1-1}{2}} - (v')^{-\f{p_1-1}{2}} | + | u^{-\f{p_1-1}{2}} - (u')^{-\f{p_1-1}{2}}|.
\end{equation}

It therefore follows that for $u,\,v\in \mathbb R$, we can define 
$\phi \restriction_{\CH_1}(u)$ and $\phi \restriction_{\CH_2}(v)$ by
$$ \phi \restriction_{\CH_1}(u) = \lim_{v\to +\infty} \phi(u,v),\qquad \phi \restriction_{\CH_2}(v) = \lim_{u\to +\infty} \phi(u,v),$$
where the limits exist because of \eqref{eq:phi.Cauchy}. Using \eqref{eq:phi.Cauchy} again, we see that
$$ \lim_{u\to +\infty} \phi \restriction_{\CH_1}(u) = \lim_{v\to +\infty} \phi \restriction_{\CH_2}(v),$$
and thus we can define
$$\phi \restriction_{\mathcal B_{\CH}} = \lim_{u\to +\infty} \phi \restriction_{\CH_1}(u) = \lim_{v\to +\infty} \phi \restriction_{\CH_2}(v).$$

We have thus defined an extension of $\phi$ to $\mathcal M_{RN} \cup \EH_{total} \cup \CH_{total}$. Finally, using again \eqref{eq:main.energy.est} again it is easy to check that the extension is continuous.

\pfstep{Step~3: Proof of continuous extendibility of $\rd_v\phi$ and $\rd_u\phi$} We show that $\rd_v\phi$ extends continuously to $\CH_2 \setminus \CH_1$; the corresponding statement for $\rd_u\phi$ can be proven in a very similar manner.

Using \eqref{eq:main.energy.est}, we have
\begin{equation}\label{eq:EE.duphi.L1}
\sup_{v\in \mathbb R} \left( \int_{-\infty}^0 (1+u^2)^{\f{p_1}2} (\rd_u\phi)^2(u,v) \, du + \int_0^{+\infty} (1+ u^2)^{\f {p_2}2} (\rd_u\phi)^2(u,v) \, du \right) < +\infty.
\end{equation}

Since $\phi$ is spherically symmetric, the wave equation $\Box_g\phi = 0$ takes the following form (cf.~\eqref{eq:ss.wave}):
\begin{equation}\label{wave.eq:rdvphi}
\rd_u (r_{RN}\rd_v\phi) = -(\rd_v r_{RN})(\rd_u\phi) = \f 14\Omg^2_{RN} \rd_u\phi,
\end{equation}
where we have used \eqref{int.r} in the second inequality.

Notice that by \eqref{eq:Omg.near.H} and \eqref{eq:Omg.near.CH}, we have the na\"ive bound $\Omg_{RN} \ls 1$. Thus \eqref{eq:EE.duphi.L1} and the Cauchy--Schwarz inequality imply that the RHS of \eqref{wave.eq:rdvphi} is $L^1$ in $u$ (uniformly in $v$). It therefore follows from \eqref{wave.eq:rdvphi} that $\rd_v\phi$ extends continuously to the Cauchy horizon $\CH_2\setminus \CH_1$.

\pfstep{Step~4: Proof of \eqref{eq:energy.est.on.CH}} By \eqref{eq:main.energy.est}, the Banach--Alaoglu theorem, and the pointwise convergence of $\phi$ as $v\to +\infty$ (established in Step~2 above), there exists a sequence $v_i \to +\infty$ such that $(\rd_u\phi)\restriction_{\{v=v_i\}}$ has a weak $L^2(w_{p_1}(-u) w_{p_2}(u)\, du)$ limit, and it is straightforward to see that this limit must coincide with the weak (and hence actual) $\partial_u$-derivative of $\phi$ along $\CH_1$. Hence, $\rd_u \phi \restriction_{\CH_1}$ satisfies
\begin{equation}\label{eq:energy.est.on.CH.1}
\|\rd_u \phi \restriction_{\CH_1} \|_{L^2(w_{p_1}(-u) w_{p_2}(u)\, du)}^2 \leq \liminf_{i\to +\infty} \| (\rd_u\phi)\restriction_{\{v=v_i\}}\|_{L^2(w_{p_1}(-u) w_{p_2}(u)\, du)}^2 \leq CA,
\end{equation}
by \eqref{eq:main.energy.est}. An entirely analogous argument gives
\begin{equation}\label{eq:energy.est.on.CH.2}
\|\rd_v \phi \restriction_{\CH_2} \|_{L^2(w_{p_1}(-u) w_{p_2}(u)\, du)}^2 \leq CA.
\end{equation}
Combining \eqref{eq:energy.est.on.CH.1} and \eqref{eq:energy.est.on.CH.2}, and then using \eqref{T.on.the.horizons}, yield \eqref{eq:energy.est.on.CH}.
\qedhere

\end{proof}

\subsection{Definition of the transmission and reflection maps}\label{sec:scat.ss}

Given Theorem~\ref{scat.SS}, we now define the \emph{transmission map} and the \emph{reflection map}. 

We prescribe $rT\phi\restriction_{\EH_1} = \Psi$ and $rT\phi\restriction_{\EH_2} = 0$, solve the wave equation, and define $\mathcal T\Psi = r T\phi\restriction_{\CH_2}$ and $\mathcal R\Psi = r T\phi\restriction_{\CH_1}$. More precisely, 
\begin{definition}\label{def:TR}
Let $p_1,\,p_2\in (1,+\infty)$. Suppose $\Psi$ is a smooth function on $\EH_1\cup \mathcal B_{\EH}$ with $\Psi\restriction_{\mathcal B_{\EH}} = 0$, $\int_{\EH_1} \Psi(v) \, dv = 0$ and 
$$\int_{\EH_1} w_{p_1}(v) w_{p_2}(-v) \Psi^2(v) \, dv <+\infty.$$

Let $\phi: \mathcal M_{RN} \cup \EH\to \mathbb R$ be the unique smooth solution to $\Box_{g_{RN}}\phi = 0$ arising from the (smooth) characteristic initial data
$$\phi \restriction_{\EH_1}(v) = \f {2}{r_+} \int_{-\infty}^v \Psi(v') \,dv' ,\quad \phi\restriction_{\EH_2}(u) \equiv 0.$$

Define the  \emph{transmission map} $\mathcal T$ and the \emph{reflection map} $\mathcal R$ by
\begin{equation}\label{eq:def.transmission.reflection}
\mathcal T \Psi  \doteq  r_- T\phi\restriction_{\CH_2},\qquad \mathcal R\Psi \doteq r_- T\phi\restriction_{\CH_1}.
\end{equation}

Note that both maps are well-defined by Theorem~\ref{scat.SS}. 
\end{definition}

\section{Statement of the main theorem for the linear wave equation}\label{sec:thm}

In this section, we give a precise statement of Theorem~\ref{thm:main.intro}. (As already mentioned in introduction, the discussion of the main theorem on mass inflation in the nonlinear setting (i.e.,~Theorem~\ref{thm:intro.MI.2}) will be postponed to Section~\ref{sec.application}.)

We give two versions of the theorem in Theorem~\ref{thm:main} and Corollary~\ref{cor:main}. The statement in Corollary~\ref{cor:main} should be thought of as the main, more important result, though it is convenient to first prove a slightly weaker statement as in Theorem~\ref{thm:main}. (The difference between the two theorems is that in Theorem~\ref{thm:main}, we make some global assumptions of the data, including requiring the data to vanish on $\mathcal B_{\EH}$ and on a large portion of $\EH_2$. In Corollary~\ref{cor:main}, we apply an easy cut-off argument to show that we only need assumptions of $\phi$ on $\EH_1$ as $v\to +\infty$.)

The conclusion of the proofs of Theorem~\ref{thm:main} and Corollary~\ref{cor:main} can be found respectively in Sections~\ref{sec.reflection} and \ref{sec:pf.cor.main}.

\begin{theorem}\label{thm:main}
Let $\phi$ be a smooth spherically symmetric solution to \eqref{wave.eqn} on $\mathcal M_{RN} \cup \mathcal H_{total}^+$. Suppose the following holds:
\begin{enumerate}
\item There exists $v_*\in \mathbb R$ such that $\phi\restriction_{\EH_1}(v) = 0$ for every $v\leq v_*$.
\item There exists $u_*\in \mathbb R$ such that $\phi \restriction_{\EH_2}(u) = 0$ for every $u \geq u_*$.
\item $\lim_{v\to +\infty} \phi\restriction_{\EH_1}(v) = 0$.
\item $\int_1^{+\infty} (1+v^2) (T\phi \restriction_{\EH_1})^2\, dv <+\infty$.
\item There exists an \underline{even} integer $p \geq 4$ which is the \underline{smallest} even integer for which
\begin{equation}\label{linear.lower.bd}
\int_1^{+\infty} (1+v^2)^{\f p2} (T\phi\restriction_{\EH_1})^2 \, dv = +\infty.
\end{equation}
\item For $p$ as above, 
\begin{equation}\label{imp.decay}
\int_1^{+\infty} (1+v^2)^{\f p2} (T^2\phi \restriction_{\EH_1})^2\, dv <+\infty.
\end{equation}
\end{enumerate}

Then the following holds (with $p$ as above):
\begin{enumerate}
\item For any $u\in \mathbb R$,
\begin{equation}\label{eq:main.instab.conclu}
\int_{1}^{+\infty} (1+v^2)^{\f p 2}(\rd_v\phi)^2(u,v)\, dv =+\infty.
\end{equation}
\item The following weighted energy along the Cauchy horizon is infinite: 
\begin{equation}\label{eq:main.CH.conclu}
\int_{-\infty}^{-1} (1+ u^2)^{\f p 2}(\lim_{v\to+\infty}(T\phi)^2(u,v))\, du =+\infty.
\end{equation}
\end{enumerate}
\end{theorem}


\begin{corollary}\label{cor:main}
Let $u_s \in (-\infty, -1)$. Suppose $\phi$ is a spherically symmetric solution to \eqref{wave.eqn} on $(\mathcal M_{RN} \cup \mathcal H^+_1) \cap \{(u,v): u\in (-\infty,u_s],\, v\in [1,+\infty)\}$ which is smooth up to $\EH_1$. Suppose that \underline{only} assumptions (3), (4), (5) and (6) of Theorem~\ref{thm:main} hold, then the following slight modifications of conclusions (1) and (2) of Theorem~\ref{thm:main} still hold:
\begin{primenumerate}
\item For any $u\in (-\infty, u_s)$,
\begin{equation}\label{eq:main.instab.conclu.prime}
\int_{1}^{+\infty} (1+v^2)^{\f p 2}(\rd_v\phi)^2(u,v)\, dv =+\infty.
\end{equation}
\item The following weighted energy along the Cauchy horizon is infinite: 
\begin{equation}\label{eq:main.CH.conclu.prime}
\int_{-\infty}^{u_s} (1+ u^2)^{\f p 2}(\lim_{v\to+\infty}(T\phi)^2(u,v))\, du =+\infty.
\end{equation}
\end{primenumerate}
\end{corollary}

\begin{remark}[Alternative assumption for the improved decay of higher derivative]
In both Theorem~\ref{thm:main} and Corollary~\ref{cor:main}, the assumption \eqref{imp.decay} can be replaced by $\exists k\geq 2$ such that 
$$\int_{1}^{+\infty} (1+v^2)^{\f p2} (T^k\phi \restriction_{\EH_1})^2\, dv <+\infty.$$ 
The proof of this more general statement is essentially the same; we omit the details.
\end{remark}

\section{The Kehle--Shlapentokh-Rothman scattering theory}\label{sec:phase.space}

In this section, we collect some facts about the transmission and reflection maps (recall Definition~\ref{def:TR}) proven in \cite{KSR}. (Some of the results in \cite{KSR} were stated in slightly different function spaces from those we consider, but their proofs can be easily adapted to our setting.) In particular, we will recall that the transmission and reflection maps admit simple phase space representations as operators defined by Fourier multipliers; see Proposition~\ref{prop:T.R.phase.rep.2}. 

\subsection{The radial ODE}

\begin{definition}
Define $V$ by 
\begin{equation}\label{eq:def.V}
V(r^*) =  \f{(1-\f{2M}{r_{RN}} + \f{{\bf e}^2}{r_{RN}^2}) [ r_{RN}(r_+ + r_-) - 2r_+ r_-]}{r_{RN}^3},
\end{equation}
where we think of $r_{RN}$ as a function of $r^*$ using \eqref{eq:r*.RN.def}.

For $\om \in \mathbb R$, define $\mathfrak u_1$ and $\mathfrak u_2$ as the unique solutions to the Volterra integral equations
\begin{equation}\label{eq:def.u1}
\mathfrak u_1(\om,r^*) = e^{i\om r^*} + \int_{-\infty}^{r^*} \f{\sin(\om(r^*-y))}{\om} V(y) \mathfrak u_1(\om,y)\, \ud y,
\end{equation}
\begin{equation}\label{eq:def.u2}
\mathfrak u_2(\om,r^*) = e^{-i\om r^*} + \int_{-\infty}^{r^*} \f{\sin(\om(r^*-y))}{\om} V(y) \mathfrak u_2(\om,y)\, \ud y;
\end{equation}
and define $\mathfrak v_1$ and $\mathfrak v_2$ as the unique solutions to the Volterra integral equations
\begin{equation}\label{eq:def.v1}
\mathfrak v_1(\om,r^*) = e^{i\om r^*} + \int_{r^*}^{+\infty} \f{\sin(\om(r^*-y))}{\om} V(y) \mathfrak v_1(\om,y)\, \ud y,
\end{equation}
\begin{equation}\label{eq:def.v2}
\mathfrak v_2(\om,r^*) = e^{-i\om r^*} + \int_{r^*}^{+\infty} \f{\sin(\om(r^*-y))}{\om} V(y) \mathfrak v_2(\om,y)\, \ud y,
\end{equation}
Here, when $\om = 0$, we define $\f{\sin(\om(r^*-y))}{\om}\restriction_{\om = 0} \doteq r^* - y$.
\end{definition}

\begin{remark}
The Volterra integral equations relate to the wave equation as follows.

The functions $\mathfrak u_1$, $\mathfrak u_2$, $\mathfrak v_1$ and $\mathfrak v_2$ are defined to satisfy the ODE
\begin{equation}\label{eq:renorm.ode}
\mathfrak u'' + (\om^2 -V) \mathfrak u =0.
\end{equation}

For sufficiently regular spherically symemtric function $\phi$, define 
$$\hat{\phi}(\om,r) = \f 1{\sqrt{2\pi}} \int_{\mathbb R} e^{-i\om t} \phi(t,r)\, dt.$$
If $\phi$ solve $\Box_g \phi = 0$, then $u = r\hat{\phi}$ satisfies \eqref{eq:renorm.ode}.

\end{remark}

\begin{definition}
Define the transmission coefficient $\mathfrak T: \mathbb R\setminus\{0\}\to \mathbb C$ and reflection coefficient $\mathfrak R:\mathbb R \setminus\{0\}\to \mathbb C$ to be the unique coefficients such that
$$\mathfrak u_1 = \mathfrak T \mathfrak v_1 + \mathfrak R \mathfrak v_2.$$

Note that they are well-defined since when $\om\neq 0$, $\mathfrak v_1$ and $\mathfrak v_2$ are linearly independent. In fact, $\mathfrak T$ and $\mathfrak R$ take the form
$$\mathfrak T(\om) = \f{\mathfrak W(\mathfrak u_1,\mathfrak v_2)}{-2i\om},\quad \mathfrak R(\om) = \f{\mathfrak W(\mathfrak u_1,\mathfrak v_1)}{2i\om},$$
where $\mathfrak W$ is the Wronskian defined by $\mathfrak W(f,g) \doteq  fg'-f'g$.
\end{definition}

We will use some properties of $\mathfrak T$ and $\mathfrak R$ proven in \cite{KSR}.
\begin{proposition}\label{prop:basic.T.R}
\begin{enumerate}
\item (Proposition~2.5 in \cite{KSR})  $\mathfrak T$ and $\mathfrak R$ extend to analytic functions on $\mathbb C\setminus \mathbb P$, where $\mathbb P = \{im\kappa_+ : m \in \mathbb N\} \cup \{ ik\kappa_-: k\in \mathbb Z\setminus \{0\}\}$ ($\kappa_{\pm}$ as in \eqref{eq:def.kappa}), are the locations of possible poles.

In particular, $\mathfrak T$ and $\mathfrak R$ are well-defined and analytic on $\mathbb R$.
\item (Theorem~2 in \cite{KSR}) $\mathfrak T$ and $\mathfrak R$ are uniformly bounded on the real line, i.e.,
$$\sup_{\om \in \mathbb R} (|\mathfrak T(\om)| + |\mathfrak R(\om)|) \ls 1.$$
\item (Proposition~2.4 in \cite{KSR}) For every $\om \in \mathbb R$,
$$|\mathfrak T(\om)|^2 - |\mathfrak R(\om)|^2 = 1.$$
\item (Proposition~2.5 in \cite{KSR}) 
$$\mathfrak T(0) = \f 12 (\f{r_-}{r_+} + \f{r_+}{r_-}),\quad \mathfrak R(0) = \f 12 (\f{r_-}{r_+} - \f{r_+}{r_-}).$$
\end{enumerate}
\end{proposition}

\begin{remark}\label{rmk:use.conservation.for.lower.bd}
Regarding points~(3) and (4) in Proposition~\ref{prop:basic.T.R}, the only thing we need below is that $\mathfrak T(0) \neq 0$ and $\mathfrak R(0)\neq 0$. Part (3) in Proposition~\ref{prop:basic.T.R} can be thought of as a (microlocalized) version of the $T$-conservation law. In particular, this means that $\mathfrak T(0) \neq 0$ follows from the $T$-conservation law as an easy consequence. On the other hand, one does not need to appeal to the conservation law, and can compute $\mathfrak T(0)$ directly, as in done is part (4) of the proposition.
\end{remark}

\subsection{The scattering map and the radial ODE}

\begin{definition}\label{def.fourier.horizon}
Let $\Psi:\EH_1\to \mathbb R$ be a spherically symmetric function (of $v$). Define the Fourier transform $\widehat{\Psi}$ (whenever well-defined) by
$$\widehat{\Psi}(\om)  \doteq  \f{1}{\sqrt{\pi}}\int_{-\infty}^\infty e^{2i \om v} \Psi(v) \, dv.$$
Similarly, for $\mathcal T\Psi:\CH_2\to \mathbb R$ and $\mathcal R\Psi:\CH_1\to \mathbb R$, we define the Fourier transform (whenever well-defined) by
$$\widehat{\mathcal T\Psi}(\om)  \doteq  \f{1}{\sqrt{\pi}}\int_{-\infty}^\infty e^{2i \om v} \widehat{\mathcal T\Psi}(v) \, dv,\quad \widehat{\mathcal R\Psi}(\om)  \doteq  \f{1}{\sqrt{\pi}}\int_{-\infty}^\infty e^{-2i \om u} \widehat{\mathcal R\Psi}(u) \, dv.$$
\end{definition}

\begin{proposition}(Theorem~3 in \cite{KSR})\label{prop:T.R.phase.rep.2}
The transmission and reflection maps defined in Definition~\ref{def:TR} are given by the following Fourier multiplier operators:
$$\widehat{\mathcal T \Psi}(\om) = \mathfrak T(\om) \widehat{\Psi}(\om),\quad \widehat{\mathcal R \Psi}(\om) = \mathfrak R(\om) \widehat{\Psi}(\om). $$
\end{proposition}

\section{Proof of the main theorem}\label{sec:proof}
In this section, we prove our main results on the linear wave equation, i.e.,~Theorem~\ref{thm:main} and Corollary~\ref{cor:main}. 

The statements (1) and (2) of Theorem~\ref{thm:main} have very similar proofs. We will give a detailed proof of the statement (1) in \textbf{Section~\ref{sec.transmission}} and then briefly discuss the necessary modifications for the statement (2) in \textbf{Section~\ref{sec.reflection}}.

The proof of Corollary~\ref{cor:main} then follows as a consequence and will be carried out in \textbf{Section~\ref{sec:pf.cor.main}}.

\subsection{The transmission map and instability results in the black hole interior}\label{sec.transmission}
In this subsection, we prove statement (1) of Theorem~\ref{thm:main}.

Before we proceed, let us give a brief summary of the argument.
\begin{enumerate}
\item First, we show that without loss of generality, we may assume that $\phi\restriction_{\mathcal{H}^+_2}(u) = 0$.
\item Next, we note that the transmission coefficient is bounded below at $\om=0$ (see Proposition~\ref{prop:T0neq0}).
\item Then, using Plancherel's theorem, we show that the assumptions of Theorem~\ref{thm:main} and the previous step to imply the blowup of a \emph{global} weighted energy at the Cauchy horizon for the transmission map (see Proposition~\ref{prop:global.blowup}).
\item Next, we argue using a physical space argument (more precisely, Theorem~\ref{scat.SS}) that the global weighted energy blows up due to the behavior as $v\to+\infty$ on $\CH_2$ (instead of $v\to -\infty$) (see Proposition~\ref{prop:local.blowup}).
\item Finally, using a local energy estimate, we show that the blowup on $\CH_2$ translates to the blow-up statement for finite $u$ stated in Theorem~\ref{thm:main} (see Proposition~\ref{prop:final.blowup}).
\end{enumerate}

We begin with the first step, which is to show that in the proof of statement (1) of Theorem~\ref{thm:main}, without loss of generality, we may take $\phi\restriction_{\mathcal{H}^+_2}(u) = 0$.
\begin{proposition} To prove Theorem~\ref{thm:main}, it suffices to prove statement (1) of the theorem under the additional assumption that $\phi\restriction_{\mathcal{H}^+_2}(v) = 0$.
\end{proposition}
\begin{proof}Let $\phi$ satisfy the assumptions of Theorem~\ref{thm:main}. We now define two auxiliary smooth solutions $\phi_{left}$ and $\phi_{right}$ to the wave equation on $\mathcal M_{RN} \cup \EH_{total}$ as follows:
\begin{enumerate}
	\item For $\phi_{left}$ we solve a characteristic initial value problem with 
	\[\phi_{left}\restriction_{\mathcal{H}^+_2}(u) = \phi\restriction_{\mathcal{H}^+_2}(u)\text{ and }\phi_{left}\restriction_{\mathcal{H}^+_1}(v) = 0.\]
	\item For $\phi_{right}$ we solve a characteristic initial value problem with 
	\[\phi_{right}\restriction_{\mathcal{H}^+_2}(u) = 0\text{ and }\phi_{right}\restriction_{\mathcal{H}^+_1}(v) = \phi\restriction_{\mathcal{H}^+_1}(v).\]
\end{enumerate}
We have, of course, that 
\begin{equation}\label{sumup}
\phi = \phi_{left} + \phi_{right}.
\end{equation}

Observe that in view of the smoothness and support assumptions of $\phi$, we have that~\eqref{eq:def.A.in.EE} holds for $\phi_{left}$ with any non-negative values of $p_1$ and $p_2$. Moreover, the other hypotheses for Theorem~\ref{scat.SS} clearly hold for $\phi_{left}$, and it is thus an immediate consequence that 
\begin{equation}\label{doesnotblowup}
\int_1^\infty (1+|v|)^{p} |\rd_v\phi_{left}|^2(u,v) \, dv < +\infty,\qquad \forall u \in \mathbb{R}.
\end{equation}
Now it suffices to observe that for any $u \in \mathbb{R}$,~\eqref{doesnotblowup} and~\eqref{sumup} imply that
\begin{equation}\label{willitblowup}
\int_1^\infty (1+|v|)^{p} |\rd_v\phi|^2(u,v) \, dv = +\infty \Leftrightarrow \int_1^\infty (1+|v|)^{p} |\rd_v\phi_{right}|^2(u,v) \, dv = +\infty.\end{equation}

\end{proof}

We now turn to the next step, which is to show that the transmission coefficient is non-zero at $\om=0$. 

\begin{proposition}\label{prop:T0neq0}
$$\mathfrak T(0)\neq 0.$$
\end{proposition}
\begin{proof}
This is immediate using either part (3) or part (4) of Theorem~\ref{prop:basic.T.R}. \qedhere
\end{proof}

Next, we use Plancherel's theorem to prove a global blow-up result.
\begin{proposition}\label{prop:global.blowup}
Let $\phi$ be as in the assumptions of Theorem~\ref{thm:main} and satisfy $\phi\restriction_{\mathcal{H}^+_2} = 0$. Define 
\begin{equation}\label{eq:def.Psi.in.theorem}
\Psi(v) = r_+ (T\phi)\restriction_{\EH_1}(v).
\end{equation}

Then
$$\int_{-\infty}^{\infty} |v|^{p}|\mathcal T \Psi|^2(v)\, dv = +\infty.$$
\end{proposition}

\begin{proof}
\pfstep{Step~1: Writing the assumptions in phase space} Conditions (1), (4) and (5) in Theorem~\ref{thm:main}, together with \eqref{eq:def.Psi.in.theorem}, imply that 
\begin{equation}\label{eq:condition.in Psi.1}
\int_{-\infty}^{\infty} |v|^{p}|\Psi|^2(v)\, dv = +\infty;
\end{equation}
and the condition (6), together with \eqref{eq:def.Psi.in.theorem}, imply that
\begin{equation}\label{eq:condition.in Psi.2}
\int_{-\infty}^{\infty} |v|^{p}|T\Psi|^2(v)\, dv< +\infty.
\end{equation}

By Plancherel's theorem (and recalling that $p$ is even), \eqref{eq:condition.in Psi.1} and \eqref{eq:condition.in Psi.2} imply
\begin{equation}\label{blow.up.H.Fourier}
\int_{-\infty}^\infty \left|\rd_\om^{\f{p}{2}}\widehat{\Psi}\right|^2 (\om)\, d\om = +\infty,\quad\int_{-\infty}^\infty \left|\rd_\om^{\f{p}{2}}(\om\widehat{\Psi})\right|^2 (\om)\, d\om< +\infty.
\end{equation}
Since $p$ is the \underline{smallest} even integer such that \eqref{linear.lower.bd} holds, we can bound
\begin{equation}\label{Psi.commutator}
\int_{-\infty}^\infty \left|[\rd_\om^{\f{p}{2}},\om]\widehat{\Psi}\right|^2 (\om)\, d\om = \f{p^2}{4} \int_{-\infty}^\infty \left|\rd_\om^{\f{p-2}{2}}\widehat{\Psi}\right|^2 (\om)\, d\om < +\infty.
\end{equation}
Combining \eqref{Psi.commutator} with the second estimate in \eqref{blow.up.H.Fourier}, we obtain
\begin{equation}\label{bound.der.hor}
\int_{-\infty}^\infty \left|\om(\rd_\om^{\f{p}{2}}\widehat{\Psi})\right|^2 (\om)\, d\om< +\infty.
\end{equation}

Combining \eqref{bound.der.hor} with the first estimate in \eqref{blow.up.H.Fourier}, we obtain
\begin{equation}\label{most.useful.est}
\int_{-1}^1 \left|\rd_\om^{\f{p}{2}}\widehat{\Psi}\right|^2 (\om)\, d\om=+\infty.
\end{equation}

\pfstep{Step~2: Finishing the proof} By Plancherel's theorem and Proposition~\ref{prop:T.R.phase.rep.2}, it suffices to prove
\begin{equation}\label{eq:blowup.intermediate}
\int_{-a}^a \left|\rd_\om^{\f{p}{2}}(\mathfrak T\widehat{\Psi})\right|^2 (\om)\, d\om= +\infty
\end{equation}
for some $a\in \mathbb R_{>0}$. We will prove \eqref{eq:blowup.intermediate} for $a=1$.

To achieve \eqref{eq:blowup.intermediate}, first notice since 
\begin{itemize}
\item $p$ is the \underline{smallest} even integer such that \eqref{linear.lower.bd} holds, and 
\item $\mathfrak T$ and its derivatives are uniformly bounded for $\om \in [-1,1]$ by analyticity (Proposition~\ref{prop:basic.T.R}), 
\end{itemize}
we have (for some constant $C>0$ depending only on $p$),
$$\int_{-a}^a \left|[\rd_\om^{\f{p}{2}},\mathfrak T]\widehat{\Psi}\right|^2 (\om)\, d\om \leq C\sum_{\substack{k_1+k_2=\f{p}{2}\\k_2\leq \f{p-2}{2}}}\int_{-1}^1 \left|(\rd_\om^{k_1}\mathfrak T)(\rd_\om^{k_2}\widehat{\Psi})\right|^2 (\om)\, d\om< +\infty.$$
It therefore suffices to establish
\begin{equation}\label{goal.inf}
\int_{-1}^1 \left|\mathfrak T(\rd_\om^{\f{p}{2}}\widehat{\Psi})\right|^2 (\om)\, d\om = +\infty.
\end{equation}
For the sake of contradiction, we assume \eqref{goal.inf} fails. Since $\mathfrak T(0)\neq 0$ (by Proposition~\ref{prop:T0neq0}) and $\mathfrak T$ is continuous (in fact even analytic by Proposition~\ref{prop:basic.T.R}), there exist $\ep \in (0,1]$ and $\eta>0$ such that $|\mathfrak T(\om)|\geq \eta$ whenever $|\om|\leq \ep$. Therefore, using this fact with \eqref{bound.der.hor} and the assumed failure of \eqref{goal.inf}, we obtain
\begin{equation*}
\begin{split}
\int_{-1}^1 \left|\rd_\om^{\f{p}{2}}\widehat{\Psi}\right|^2 (\om)\, d\om
\leq \eta^{-2}\int_{-\ep}^\ep \left|\mathfrak T(\rd_\om^{\f{p}{2}}\widehat{\Psi})\right|^2 (\om)\, d\om+\ep^{-2}\int_{[-1,1]\setminus [-\ep,\ep]} \left|\om(\rd_\om^{\f{p}{2}}\widehat{\Psi})\right|^2 (\om)\, d\om< +\infty.
\end{split}
\end{equation*}
This contradicts \eqref{most.useful.est}. Therefore, \eqref{goal.inf} holds.
\end{proof}

At this point, the blow-up result that we obtained in Proposition~\ref{prop:global.blowup} is global, i.e.,~we do not yet know whether the blowup is due to the behavior of $\mathcal T\Psi$ as $v\to +\infty$ or as $v\to -\infty$. Nevertheless, in the following proposition, we are able to use the energy estimates in Section~\ref{sec:scat.ss} to deduce that the blowup is due to a lower bound of the decay as $v\to +\infty$.

\begin{proposition}\label{prop:local.blowup}
Let $\Psi$ be as in the assumptions of Proposition~\ref{prop:global.blowup}. Then $\mathcal T \Psi$ satisfies
$$\int_{-\infty}^{1} (1+|v|)^{p}|\mathcal T \Psi|^2(v)\, dv < +\infty.$$
Therefore, when combined with Proposition~\ref{prop:global.blowup}, we obtain
$$\int_{1}^{\infty} (1+|v|)^{p}|\mathcal T \Psi|^2(v)\, dv = +\infty.$$
\end{proposition}
\begin{proof}
This is an immediate corollary of Theorem~\ref{scat.SS} with $p_1 = 2$ and $p_2 = p$. Notice that the assumptions (1)--(4) of Theorem~\ref{thm:main}, as well as the additional assumption $\phi \restriction_{\calH_{2}^{+}} = 0$ as in Proposition~\ref{prop:global.blowup}, imply that the assumptions (1)--(2) of Theorem~\ref{scat.SS} are satisfied. 
\qedhere
\end{proof}

\begin{proposition}\label{prop:final.blowup}
Let $\phi$ satisfy the assumptions of Theorem~\ref{thm:main} and satisfy $\phi\restriction_{\mathcal{H}^+_2} = 0$. Then the following holds for every $u\in \mathbb R$:
\begin{equation}\label{blow.up.local}
\int_1^\infty (1+|v|)^{p} |\rd_v\phi|^2(u,v) \, dv = +\infty.
\end{equation}
\end{proposition}
\begin{proof}
To show that \eqref{blow.up.local} also holds for all finite $u$, we will argue by contradiction. Assuming that \eqref{blow.up.local} fails for some $u_i \in \mathbb R$, we use a standard propagation of regularity argument with \emph{finite time} energy estimates to arrive at a contradiction with the conclusion of Proposition~\ref{prop:local.blowup}.

Assume for the sake of contradiction that $\exists u_i\in \mathbb R$ such that
\begin{equation}\label{localEE.contra}
\int_1^{\infty} (1+|v|)^{p}(\rd_v\phi)^2(u_i,v)\, dv < +\infty.
\end{equation}

\pfstep{Step~0: Some preliminary estimates for $\phi$}
Before we proceed, first note that due to the conditions (1)--(4) of Theorem~\ref{thm:main}, we can apply Theorem~\ref{scat.SS} with $p_1 = 2$ and $p_2 = p$. In particular, \eqref{eq:main.energy.est} implies that
\begin{equation}\label{eq:localEE.nonlinear.duphi.L1}
\sup_{v\in \mathbb R} \left( \int_{-\infty}^{-1} (1+u^2) (\rd_u\phi)^2(u,v) \, du + \int_1^{+\infty} (1+ u^2)^{\f p2} (\rd_u\phi)^2(u,v) \, du \right) < +\infty.
\end{equation}

\pfstep{Step~1: Propagation of regularity by finite time energy estimates} Multiplying \eqref{wave.eq:rdvphi} by $(1+|v|)^{p} r_{RN}\rd_v\phi$, integrating a region $(u,v)\in [u_i,u_f]\times [0,v_f]$ and integrating by parts, we obtain
\begin{equation}\label{localEE}
\begin{split}
&\:\f 12 \int_1^{v_f} (1+|v|)^{p} r_{RN}^2(\rd_v\phi)^2(u_f,v)\, dv \\
=&\: \f 12 \int_1^{v_f} (1+|v|)^{p} r_{RN}^2(\rd_v\phi)^2(u_i,v)\, dv + \f 14 \int_1^{v_f}\int_{u_i}^{u_f} (1+|v|)^{p} \Omg_{RN}^2 r_{RN} (\rd_u\phi)(\rd_v\phi)(u,v) \,du\,dv.
\end{split}
\end{equation}
The last term on the RHS of \eqref{localEE} can be controlled as follows using H\"older's inequality:
\begin{equation*}
\begin{split}
&\: \left|\int_1^{v_f}\int_{u_i}^{u_f} (1+|v|)^{p} \Omg_{RN}^2 r_{RN} (\rd_u\phi)(\rd_v\phi)(u,v) \,du\,dv \right|\\
\leq &\: r_+ \left(\int_1^{v_f}\int_{u_i}^{u_f} (1+|v|)^{p} \Omg_{RN}^2|\rd_v\phi|^2(u,v) \,du\,dv\right)^{\f 12}\left(\int_1^{v_f}\int_{u_i}^{u_f} (1+|v|)^{p} \Omg_{RN}^2 |\rd_u\phi|^2(u,v) \,du\,dv\right)^{\f 12}\\
\ls &\: \left(\sup_{u\in [u_i,u_f)}\int_{1}^{v_f} (1+|v|)^{p}|\rd_v\phi|^2(u,v)\, dv\right)^{\f 12} \underbrace{\left(\int_{u_i}^{u_f} \sup_{v\in [1,v_f)}\Omg_{RN}^2(u,v)\, du\right)^{\f 12}}_{\doteq I}\\
&\:\times \underbrace{\left(\sup_{v\in [1,v_f)}\int_{u_i}^{u_f} |\rd_u\phi|^2(u,v)\, du\right)^{\f 12}}_{\doteq II} \underbrace{\left(\int_{1}^{v_f} \sup_{u\in [u_i,u_f)}(1+|v|)^{p}\Omg_{RN}^2(u,v)\, dv\right)^{\f 12}}_{\doteq III}.
\end{split}
\end{equation*}
Note that in this region (which is a subset of $\{(u,v): u\geq u_i,\, v\geq 0\}$), we have $\sup_{v\in [0,+\infty)} \Omg_{RN}^2(u,v) \ls \min\{e^{- 2\kappa_- u}, 1\}$ and $\sup_{u\in [u_i,+\infty)} \Omg_{RN}^2 \ls e^{-2 \kappa_- v}$ by \eqref{eq:Omg.near.CH}. Thus, $I$ and $III$ are bounded uniformly for all $(u_f,v_f)\in [u_i,\infty)\times [0,\infty)$. 

Moreover, $II$ is uniformly bounded for all $(u_f,v_f)\in [u_i,\infty)\times [1,\infty)$ thanks to \eqref{eq:localEE.nonlinear.duphi.L1}.

Therefore, using \eqref{localEE.contra}, \eqref{localEE} and Young's inequality, we conclude that $\f 12 \int_1^{v_f} (1+|v|)^{p} r^2(\rd_v\phi)^2(u_f,v)\, dv$ is uniformly bounded for all $(u_f,v_f)\in [u_i,\infty)\times [0,\infty)$. Recall now from part (3) of Theorem~\ref{scat.SS} that the pointwise limit $\lim_{u\to\infty}|\rd_v\phi|^2(u,v)$ exists. By Fatou's lemma, we can therefore take $u_f,\,v_f\to +\infty$ so as to show that 
\begin{equation}\label{eq:to.contradict.the.global.blowup}
\f 12 \int_1^{\infty} (1+|v|)^{p} r_{RN}^2(\lim_{u\to\infty}|\rd_v\phi|^2(u,v))\, dv < +\infty.
\end{equation}

On the other hand, by the definition of $\mathcal T$ in \eqref{eq:def.transmission.reflection}, the estimate in Proposition~\ref{prop:local.blowup}, and \eqref{T.on.the.horizons},
\begin{equation}\label{blow.up.global}
\int_1^\infty (1+|v|)^{p} r_{RN}^2 (\lim_{u\to +\infty}|\rd_v\phi|^2(u,v)) \, dv = +\infty.
\end{equation}
Obviously, \eqref{eq:to.contradict.the.global.blowup} and \eqref{blow.up.global} contradict each other. It follows that \eqref{localEE.contra} does not hold. \qedhere
\end{proof}

\subsection{The reflection map and power-law tails along the Cauchy horizon}\label{sec.reflection}

We now turn to the proof of the second statement of Theorem~\ref{thm:main}. Inspecting the argument for statement (1) of Theorem~\ref{thm:main} in Section~\ref{sec.transmission}, one sees that as long as we can show $\mathfrak R(0)\neq 0$, the remainder of the argument proceeds in an identical manner. 

The fact that $\mathfrak R(0)\neq 0$ is an immediate consequence of part (4) of Proposition~\ref{prop:basic.T.R}. (Notice that unlike the corresponding statement for $\mathfrak T(0)$, the conservation law in part (3) of Proposition~\ref{prop:basic.T.R} does not provide any lower bound for $\mathfrak R$.) Once we know that $\mathfrak R(0)\neq 0$, the remaining steps of the proof proceed in exactly the same manner as Propositions~\ref{prop:global.blowup} and \ref{prop:local.blowup}; we omit the details. (Note that in this case, the analogue of Proposition~\ref{prop:local.blowup} already gives the desired result and we do not need to consider an analogue of Proposition~\ref{prop:final.blowup}.) With this, we also conclude the proof of Theorem~\ref{thm:main}.

\subsection{Proof of Corollary~\ref{cor:main}}\label{sec:pf.cor.main}

\begin{proof}[Proof of Corollary~\ref{cor:main}]
We prove Corollary~\ref{cor:main} by reducing it to Theorem~\ref{thm:main}. The main difference between Corollary~\ref{cor:main} and Theorem~\ref{thm:main} is that in Corollary~\ref{cor:main}, the solution is only defined on a subset of $\mathcal M_{RN}$. In particular, Corollary~\ref{cor:main} does not impose the support properties assumptions (1)--(2) in Theorem~\ref{thm:main}. We will therefore extend $\phi$ to all of $\mathcal M_{RN}$ and then use the finite speed of propagation to reduce to Theorem~\ref{thm:main}

More precisely, given $\phi$ satisfying only the assumptions of Corollary~\ref{cor:main}, we want to construct a solution $\widetilde{\phi}$ to $\Box_{g_{RN}} \widetilde{\phi} = 0$ which is smooth on $\mathcal M_{RN} \cup \EH_{total}$ and satisfies
\begin{align}
	\widetilde{\phi}(u,v) = \phi(u,v) & \mbox{\quad when $u\leq u_s$ and $v\geq 1$}, \label{eq:tildephi.1}\\
	\widetilde{\phi}\restriction_{\EH_1}(v) = 0 & \mbox{\quad when $v \leq 0$}, \label{eq:tildephi.2}\\
	\widetilde{\phi}\restriction_{\EH_2}(u) = 0 & \mbox{\quad when $u \geq 0$}. \label{eq:tildephi.3}
\end{align}

We complete the proof of the corollary assuming for the moment that such a $\widetilde{\phi}$ exists. By \eqref{eq:tildephi.1}, assumptions (3)--(6) of Theorem~\ref{thm:main} are satisfied by $\widetilde{\phi}$ (since the asymptotics are not changed at all). By \eqref{eq:tildephi.2}, assumption (1) of Theorem~\ref{thm:main} is satisfied by $\widetilde{\phi}$; while by \eqref{eq:tildephi.3}, assumption (2) of Theorem~\ref{thm:main} is satisfied by $\widetilde{\phi}$. It therefore follows that Theorem~\ref{thm:main} can be applied to $\widetilde{\phi}$ so that 
\eqref{eq:main.instab.conclu} and \eqref{eq:main.CH.conclu} both hold, but for $\widetilde{\phi}$ instead of $\phi$. Using \eqref{eq:tildephi.1} again, this then implies \eqref{eq:main.instab.conclu} and \eqref{eq:main.CH.conclu} both hold for $\phi$.

It remains to construct a smooth solution $\widetilde{\phi}$ to $\Box_{g_{RN}}\widetilde{\phi} = 0$ satisfying \eqref{eq:tildephi.1}--\eqref{eq:tildephi.3}. This is achieved in the following two steps.

\pfstep{Step~1: Construction of $\widetilde{\phi}$ in $\{(u,v): u \leq u_s,\, v\in \mathbb R\}$} We first define, on the event horizon $\EH_1$, $\widetilde{\phi}\restriction_{\EH_1}$ to be a smooth function so that $\widetilde{\phi}\restriction_{\EH_1}(v) = {\phi}\restriction_{\EH_1}(v)$ for $v \geq 1$ and \eqref{eq:tildephi.2} holds. Define also $\widetilde{\phi}$ on the line $\{(u,v): u\leq u_s, v = 1\}$ so that $\widetilde{\phi}\restriction_{\{(u,v): u\leq u_s, v = 1\}} = \phi\restriction_{\{(u,v): u\leq u_s,\, v = 1\}}$.

We now solve the characteristic initial problem for $\Box_{g_{RN}}\widetilde{\phi} = 0$ with data imposed above. First, by solving the characteristic initial value problem, we define a solution in $\{(u,v): u\leq u_s,\,v\geq 1\}$, which obeys \eqref{eq:tildephi.1} by a domain of dependence argument. Then, we solve again the characteristic initial value problem\footnote{We used spherical symmetry here to show that this characteristic initial value problem is well-posed. However, notice that Corollary~\ref{cor:main} can also be proven with an argument that generalizes outside spherical symmetry. For instance, after introducing appropriate cut-offs, we can show using energy estimates that the cutoff error only gives finite contributions to \eqref{eq:main.instab.conclu.prime} and \eqref{eq:main.CH.conclu.prime}.}, but now with data on $\{(u,v): u\leq u_s,\, v=1\}$ and $\{(u,v): u=-\infty,v\leq 1\}$ so that we obtain a solution in the region $\{(u,v): u\leq u_s,\,v\leq 1\}$ which contains $\mathcal B_{\EH}$. 

We have thus constructed $\widetilde{\phi}$ in $\{(u,v): u \leq u_s,\, v\in \mathbb R\}$. Importantly, \eqref{eq:tildephi.2} is satisfied by definition of the initial data, and $\widetilde{\phi}$ is smooth in $\{(u,v): u \leq u_s,\, v\in \mathbb R\}$ up to the event horizon.

\pfstep{Step 2: Construction of $\widetilde{\phi}$ in $\{(u,v): u \geq u_s,\, v\in \mathbb R\}$} To define $\widetilde{\phi}$ in the remaining region, we again solve a characteristic initial value problem. First, we take $\widetilde{\phi}\restriction_{\{(u,v): u=u_s\}}$ to be as given by the construction in Step~1. Define also $\widetilde{\phi} \restriction_{\EH_2\cap \{(u,v): u\geq u_s\} }$ so that $\widetilde{\phi} \restriction_{\EH_2}$ is smooth and $\widetilde{\phi} \restriction_{\EH_2}(u) = 0$ for $u\geq 0$.

Next, we solve the characteristic initial value problem with the data given above to obtain a solution in the region $\{(u,v): u \geq u_s,\, v\in \mathbb R\}$. The resulting $\widetilde{\phi}$ is therefore a solution to the wave equation and is smooth in $\mathcal M_{RN} \cup \EH_{total}$. By definition of the data, \eqref{eq:tildephi.3} is satisfied. Combining this with Step~1, in which we proved \eqref{eq:tildephi.1}--\eqref{eq:tildephi.2}, we have thus concluded the construction. \qedhere
\end{proof}

\section{Application to mass inflation for the Einstein--Maxwell--scalar field system in spherical symmetry}\label{sec.application}

In this section, we apply our Theorem~\ref{thm:main} to the problem of \textbf{mass inflation} for the Einstein--Maxwell--scalar field system \eqref{EMSFS} in spherical symmetry. This proves Theorem~\ref{thm:intro.MI.2}. See Theorem~\ref{mass.infl.thm} in Section~\ref{sec:main.result.mass.inflation} for the precise statement that we prove.

We will first recall the definition and various facts about the Einstein--Maxwell--(real) scalar field system in spherical symmetry in \textbf{Section~\ref{sec.SS}}. In \textbf{Section~\ref{sec:the.old.LO.results}}, we then recall the results established in \cite{LO.interior, LO.exterior}. In \textbf{Section~\ref{sec:main.result.mass.inflation}}, we state our main result on mass inflation (Theorem~\ref{mass.infl.thm}). In \textbf{Section~\ref{sec:MI.crit}}, we begin the proof of Theorem~\ref{mass.infl.thm} by proving a mass inflation criterion which reduces the proof of Theorem~\ref{mass.infl.thm} to showing that the scalar field is not identically $0$ at the Cauchy horizon. In \textbf{Section~\ref{sec:EE.inhom}}, we establish general energy estimates for inhomogeneous wave equations on Reissner--Nordstr\"om. This in particular allows us to use our main linear result (part (2) of Corollary~\ref{cor:main}) in a perturbative argument. Finally, in \textbf{Section~\ref{sec:proofofmassinflation}}, we conclude the proof of Theorem~\ref{mass.infl.thm}. 

\subsection{Einstein--Maxwell--(real) scalar field system in spherical symmetry}\label{sec.SS}

We first discuss spherically symmetric solutions to the Einstein--Maxwell--(real) scalar field system \eqref{EMSFS}. The following definition is directly from \cite{LO.interior}.
\begin{definition}[Spherically symmetric solutions]\label{def.SS}
Let $(\mathcal M,g,\phi,F)$ be a suitably regular\footnote{The precise regularity is irrelevant here. For relevant well-posedness statements, see \cite[Propositionss~2.4, 2.5]{LO.interior}.} solution to the Einstein--Maxwell--(real)--scalar--field system \eqref{EMSFS}. We say that $(\mathcal M,g,\phi,F)$ is \emph{spherically symmetric} if the following properties hold:
\begin{enumerate}
\item The symmetry group $SO(3)$ acts on $(\mathcal M,g)$ by isometry with spacelike orbits.
\item The metric $g$ on $\mathcal M$ is given by
\begin{equation}\label{SS.metric.1}
g=g_{\mathcal Q}+r^2 g_{\mathbb S^2},
\end{equation}
where
\begin{equation}\label{SS.metric.2}
g_{\mathcal Q}=-\f{\Omg^2 }{2}(du\otimes dv+dv\otimes du)
\end{equation}
is a Lorentzian metric on the $2$-dimensional manifold $\mathcal Q=\mathcal M/SO(3)$ and $r$ is defined to be the area radius function of the group orbit, i.e.,
$$r=\sqrt{\f{\mbox{Area}({\boldsymbol \pi}^{-1}(p))}{4\pi}},$$
for every $p\in \mathcal Q$, where ${\boldsymbol \pi}$ is natural projection ${\boldsymbol \pi}:\mathcal M\to \mathcal Q$ taking a point to the group orbit it belongs to. Here, as in the introduction, $g_{\mathbb S^2}$ denotes the standard round metric on $\mathbb S^2$ with radius $1$.
\item The function $\phi$ at a point $x$ depends only on ${\boldsymbol \pi}(x)$, i.e.,~for $p\in \mathcal Q$ and $x,y\in {\boldsymbol \pi}^{-1}(p)$, it holds that $\phi(x)=\phi(y)$.
\item The Maxwell field $F$ is invariant under pullback by the action (by isometry) of $SO(3)$ on $\calM$.
Moreover, there exists ${\bf e}:\mathcal Q\to \mathbb R$ such that
$$F=\f{\bfe}{2({\boldsymbol \pi}^* r)^2}{\boldsymbol \pi}^*(\Omg^2\,du\wedge dv).$$
\end{enumerate}
\end{definition}

It is well-known that for a solution to \eqref{EMSFS}, ${\bf e}$ is in fact a constant.

In spherical symmetry, the Einstein-Maxwell-(real)-scalar-field system reduces to the following system of coupled wave equations for $(r,\phi,\Omega)$
\begin{equation}\label{WW.SS}
\begin{cases}
\rd_u\rd_v r=-\f {\Omg^2}{4 r}-\f{\rd_u r\rd_v r}{r}+\f{\Omg^2 {\bf e}^2}{4r^3},\\
\rd_u\rd_v \phi=-\f{\rd_v r\rd_u\phi}{r}-\f{\rd_u r\rd_v\phi}{r},\\
\rd_u\rd_v\log\Omg=-\rd_u\phi\rd_v\phi-\f {\Omg^2 {\bf e}^2}{2 r^4}+\f{\Omg^2}{4r^2}+\f{\rd_u r\rd_v r}{r^2}.
\end{cases}
\end{equation}

The solution moreover satisfy the following Raychaudhuri equations:
\begin{equation}\label{eqn.Ray}
\begin{cases}
\rd_v(\f{\rd_v r}{\Omg^2})=-\f{r(\rd_v\phi)^2}{\Omg^2},\\
\rd_u(\f{\rd_u r}{\Omg^2})=-\f{r(\rd_u\phi)^2}{\Omg^2}.
\end{cases}
\end{equation}
In the context of the characteristic initial value problem, it can be easily shown that if \eqref{eqn.Ray} are initially satisfied, then they are propagated by \eqref{WW.SS}.

We will often use the following short hand for $\rd_v r$ and $\rd_u r$:
$$\lambda \doteq \rd_v r,\quad \nu \doteq \rd_u r.$$

Define the \emph{Hawking mass} $m:\mathcal Q\to \mathbb R$ by
$$m \doteq \f r2(1-g_{\mathcal Q}(\nabla r,\nabla r)),$$
where $g_{\mathcal Q}$ is as defined in \eqref{SS.metric.2}. Alternatively, we can write
$$m = \f r 2(1-\f{4\rd_u r\rd_v r}{\Omg^2}).$$
Finally, define the \emph{renormalized Hawking mass} $\varpi:\mathcal Q\to \mathbb R$ by 
\begin{equation}\label{eq:renormalized}
\varpi = m + \f{{\bf e}}{r^2}.
\end{equation}
Combining \eqref{WW.SS} and \eqref{eqn.Ray} (see also \cite[equation (2.8)]{LO.interior}), it is easy to deduce that
\begin{equation}\label{eq:varpi}
\rd_u\varpi=-\f{2 r^2 \rd_v r(\rd_u\phi)^2}{\Omg^2}, \quad \rd_v\varpi=-\f{2 r^2 \rd_u r(\rd_v\phi)^2}{\Omg^2}.
\end{equation}

\subsection{Results in \cite{LO.interior, LO.exterior}}\label{sec:the.old.LO.results}

In this subsection, we recall some results in the earlier works \cite{LO.interior, LO.exterior}. These works show that generic solutions are $C^2$-future inextendible. However, these works by themselves do not guarantee that mass inflation holds.

For expositional purposes, let us only focus on the case where the initial data have compactly supported scalar fields. In the works \cite{LO.interior, LO.exterior}, one can also allow polynomially decaying but non-compactly supported initial scalar fields. It is straightforward to obtain modifications of our Theorem~\ref{mass.infl.thm} in that setting; we omit the details.

In \cite{LO.interior}, we introduced a notion of \emph{generic} two-ended asymptotically flat future-admissible spherically symmetric Cauchy data for \eqref{EMSFS}. Since the precise definition is largely irrelevant to the remainder of the paper, we will not repeat the definitions and refer the reader to \cite[Sections~3.1--3.3]{LO.interior} instead. For the purpose of this discussion, let us just denote by $\mathcal G_c$ the set of generic data introduced in \cite{LO.interior} which moreover \emph{are smooth and have compactly supported initial scalar field}. The theorems that we cite below apply in particular to solutions arising from initial data in $\mathcal G_c$.

\subsubsection{A priori boundary characterization} We first state a preliminary result regarding the a priori boundary characterization of solutions. This allows us to talk about various regions and boundary components of the solutions. In the statement of the following theorem, we will only refer to the a priori boundary characterization in terms of the corresponding Penrose diagram. To make precise all the relevant notions needed for the Penrose diagram will take us too far afield. Instead, we refer the reader to \cite[Theorem~4.1]{LO.interior}.
\begin{theorem}[Persistence of the Cauchy horizon (Dafermos \cite{D2}, Dafermos--Rodnianski \cite{DRPL}) and the boundary characterization of the solution (Dafermos \cite{D3}, Kommemi \cite{Kommemi})]\label{thm:Kommemi}
Consider a two-ended asymptotically flat future-admissible spherically symmetric initial data set in $\mathcal G_c$ and let $(\calM, g, \phi, F)$ be the corresponding maximal globally hyperbolic future development.

Then $(\calM, g)$ has one of the following two Penrose diagrammatic representations:
\begin{figure}[h] \label{fig:main.structure}
\begin{center}
\def\svgwidth{450px}
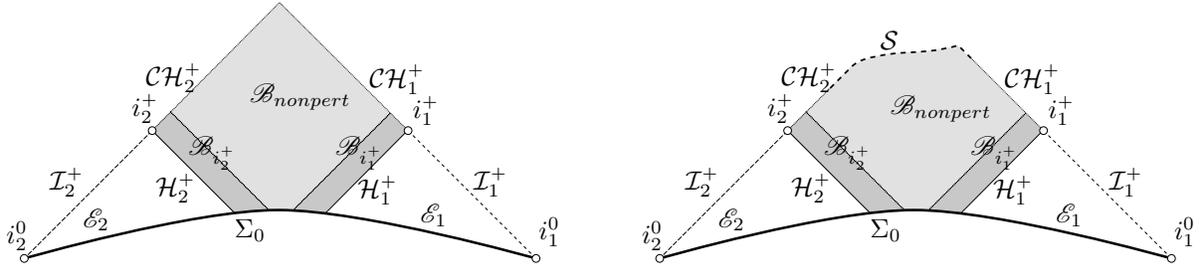 
\caption{Penrose diagram of the maximal globally hyperbolic future development of future admissible initial data} 
\end{center}
\end{figure}
\end{theorem}

We note already in terms of the boundary components defined above, \emph{mass inflation} will mean that $m=+\infty$ identically on both $\CH_1$ and $\CH_2$. Since the situation is completely symmetric for $\CH_1$ and $\CH_2$, from now on, we focus our discussion on $\CH_1$. Completely analogous statements also hold for $\CH_2$.

\subsubsection{Stability of the Cauchy horizon} Already implicit in Theorem~\ref{thm:Kommemi} is a statement (proven in \cite{D2}) that solutions arising from generic initial data set in $\mathcal G_c$ must have ``a piece'' of (null) Cauchy horizon. In fact, in the ``rectangular strip'' $\mathscr B_{i^+_1}$ near the event horizon $\EH_1$ (in either diagram) in Figure~1, the solution approaches Reissner--Nordstr\"om toward $i_1^+$ with a quantitative inverse polynomial rate. 

To summarize this stability result, let us first introduce the following change of coordinates.
\begin{definition}\label{def:u.v.system.from.U.v}
Given a system of null coordinates $(U,v)$ in $\{(U,v): 0\leq U\leq U_s,\, 1\leq v\leq +\infty\}$, construct a new system of coordinates $(u,v)$ by defining a new coordinate function $u = u(U)$ so that the inverse function $U = U(u)$ satisfies the following (compare \eqref{U.EH.def})
\begin{equation}\label{U.EH.def.nonlinear}
\f{dU}{du}=e^{2\kappa_+ u} \mbox{ and }\, U(u)\to 0 \mbox{ as }u\to -\infty,
\end{equation}
where $\kpp_{+} = \frac{r_{+}-r_{-}}{2 r_{+}^{2}}$ and $r_{\pm} = M \pm \sqrt{M^{2}-\bfe^{2}}$ as in \eqref{eq:def.kappa}, and $M$ and $\bfe$ shall be specified when the definition is applied.
Moreover, define $u_s$ by
\begin{equation}\label{eq:def.us}
u_s \doteq u(U_s).
\end{equation}
\end{definition}

The following theorem can be obtained by combining Theorems~4.4~and~5.1 in \cite{LO.interior}.

\begin{theorem}[Stability of the Cauchy horizon]\label{main.theorem.C0.stability}
Consider a two-ended asymptotically flat future-admissible spherically symmetric initial data set in $\mathcal G_c$ and let $(\calM, g, \phi, F)$ be the corresponding maximal globally hyperbolic future development.

Then there exist 
\begin{itemize}
\item constants $M$ and ${\bf e}$ with $0<|{\bf e}|<M$, 
\item a double null coordinate system $(U,v)$ on $\mathcal Q = \mathcal M/SO(3)$, and
\item a spacetime region $\{ (U,v): 0\leq U \leq U_s,\, 1\leq v < +\infty\}$ (in a double null coordinate system above) in the black hole interior
\end{itemize} 
such that the solution settles down to the Reissner--Nordstr\"om interior with parameters $M$ and ${\bf e}$.

More precisely, after
\begin{itemize} 
\item introducing a system of double null coordinate $(u,v)$ as in Definition~\ref{def:u.v.system.from.U.v} with the above values of $M$ and $\bfe$, (and assuming that $u_s\leq -1$), and
\item letting $r_{RN}$ and $\Omg_{RN}$ denote the functions of $(u,v)$ corresponding to the metric components of the Reissner--Nordstr\"om interior of parameters $M$ and ${\bf e}$ in the coordinate system in Section~\ref{sec.null.1},
\end{itemize}
the following estimates hold for every $\rho <3$ in the $(u,v)$ coordinate system in $\{(u,v): -\infty< u< u_s,\, 1\leq v<+\infty\}$, for some constant $C>0$ (depending on $\rho$):
$$|\phi|(u,v)+|r-r_{RN}|(u,v)+|\log\Omg-\log\Omg_{RN}|(u,v) \leq C(v^{-\rho}+|u|^{-\rho+1}),$$
$$|\rd_v\phi|(u,v)+|\rd_v(r-r_{RN})|(u,v)+|\rd_v(\log\Omg-\log\Omg_{RN})|(u,v) \leq Cv^{-\rho}.$$
Furthermore, for every $A\in \mathbb R$, there exists $C>0$ depending on $A$ and $\rho$ such that the following estimates hold $\{(u,v): -\infty< u< u_s,\, 1\leq v<+\infty\}$:
$$|\rd_u\phi|(u,v)+|\rd_u(r-r_{RN})|(u,v)+|\rd_u(\log\Omg-\log\Omg_{RN})|(u,v) \leq 
\begin{cases}
C\Omg_{RN}^2 v^{-\rho} &\mbox{for }u+v\leq A\\
C|u|^{-\rho} &\mbox{for }u+v\geq A.
\end{cases}
$$
\end{theorem}

\subsubsection{Lower bound of the scalar field along the event horizon} Next, we state a lower bound for the scalar field along the event horizon, which holds for solutions arising from a generic initial set in $\mathcal G_c$; see Theorem~\ref{final.blow.up.step}. This result requires a genuine genericity condition, in addition to just requiring the charge ${\bf e} \neq 0$. 

The lower bound in Theorem~\ref{final.blow.up.step} is intimately connected to the blowup in the black hole interior and the generic $C^2$-inextendibility proven in \cite{LO.interior, LO.exterior}.

\begin{theorem}[Lower bound along the event horizon]\label{final.blow.up.step}
Consider a two-ended asymptotically flat future-admissible spherically symmetric initial data set in $\mathcal G_c$ and let $(\calM, g, \phi, F)$ be the corresponding maximal globally hyperbolic future development. 
Then for an advanced null coordinate $v$ such that
\begin{equation*}
	C^{-1} < \inf_{\EH_1} \frac{\rd_{v} r}{1-\frac{2m}{r}} \leq \sup_{\EH_1} \frac{\rd_{v} r}{1-\frac{2m}{r}} < C
\end{equation*}
for some $C > 0$, we have\footnote{Note that in \cite{LO.interior}, we have the stronger result that for any $\alp>7$,
$$\int_{\EH_1} v^{\alp} (\rd_{v} \phi)^{2} \, \ud v = +\infty.$$
We will not need this stronger statement.
}
\begin{equation}\label{eq:mass.inflation.lower.bd}
	\int_{\EH_1} v^{8} (\rd_{v} \phi)^{2} \, \ud v = +\infty.
\end{equation}
\end{theorem}

\subsubsection{Blow-up at $\CH_1$}\label{sec:CH.blowup.from.old.paper}

In order to describe the result concerning blow-up in the black hole interior, we introduce yet another double null coordinate system. 
\begin{definition}\label{def:uV}
We introduce the coordinate as follows. 
\begin{enumerate}
\item We start with the coordinate system $(U,v)$ as in Theorem~\ref{main.theorem.C0.stability}, except that $U$ is defined for the whole $\CH_1$, beyond the perturbative region $U\leq U_s$. 
\item Given the coordinate system $(U,v)$, define $(u,v)$ as in Definition~\ref{def:u.v.system.from.U.v}.
\item Define $V = V(v)$ by the following relation
\begin{equation}\label{V.CH.def.1}
\f{dV}{dv}=e^{-2\kappa_- v}\mbox{ and }V(v)\to 1\mbox{ as }v\to +\infty.
\end{equation}
where $\kpp_{-} = \frac{r_{+}-r_{-}}{2 r_{-}^{2}}$ and $r_{\pm} = M \pm \sqrt{M^{2}-\bfe^{2}}$ as in \eqref{eq:def.kappa}, and $M$ and $\bfe$ are as in Theorem~\ref{main.theorem.C0.stability}
\end{enumerate}
Moreover, 
\begin{enumerate}
\item define $V_1 \doteq V(1)$, and
\item define $u_{\CH_1} \in (-\infty,\infty]$ to be such that $p_{\CH_1} = (u_{\CH_1}, 1) \in \mathcal Q$ corresponds to the future end-point of $\CH_1$ in the $(u,V)$ coordinates.
\end{enumerate}
\end{definition}

In this subsection $\Omg$ is taken to be the metric component in $(u,V)$ coordinates, i.e.,~$g_{\mathcal Q} = -\f{\Omg^2}2 (du \otimes dV + dV \otimes du)$.

The following theorem can be found in \cite[Theorem~5.5]{LO.interior}, specialized to $\mathcal G_c$. This is a global theorem concerning $\CH_1$, even beyond the perturbative region considered in Theorem~\ref{main.theorem.C0.stability}.
\begin{theorem} \label{thm.nonpert}
Consider a two-ended asymptotically flat future-admissible spherically symmetric initial data set in $\mathcal G_c$ and let $(\calM, g, \phi, F)$ be the corresponding maximal globally hyperbolic future development.  

In a neighborhood of $\EH_1$ in the interior of the black hole, consider the null coordinates $(u, V)$ defined as in Definition~\ref{def:uV}. Then the metric components $\Omg^{2}(u, V)$ and $r(u, V)$, as well as the scalar field $\phi(u, V)$, extend continuously to $\CH_{1} \setminus \set{p_{\CH_{1}}} = \set{(u, V) : -\infty < u < u_{\CH_{1}}, V = 1}$. The extended metric components $\Omg^{2}(u, V)$ and $r(u, V)$ are nonvanishing on $\CH_{1} \setminus \set{p_{\CH_{1}}}$.

Moreover, if the lower bound \eqref{eq:mass.inflation.lower.bd} holds on $\EH_{1}$, then for every $u \in (-\infty, u_{\CH_{1}})$, the following blow up of $\rd_{V} \phi$ and $\rd_{V} r$ hold:
\begin{align} 
	\int_{0}^{1} \frac{(\rd_{V} \phi)^{2}}{\Omg^{2}}(u, V) d V = \infty, \label{eq:nonpert-blowup-phi} \\
	\lim_{V \to 1} \frac{\rd_{V} r}{\Omg^{2}}(u, V) = - \infty.		\label{eq:nonpert-blowup-dvr}
\end{align}
In particular, the scalar field is not in $W^{1,2}_{loc}$ and the metric is not in $C^1$ in the above $C^0$ extension obtained by adjoining $\set{(u, V) : -\infty < u < u_{\CH_{1}}, V = 1}$.
\end{theorem}

In the following, only \eqref{eq:nonpert-blowup-dvr} will be used in the proof of our generic mass inflation result.

\subsection{Main result on mass inflation}\label{sec:main.result.mass.inflation}

The following is our main theorem on mass inflation for the Einstein--Maxwell--(real)--scalar--field system in spherical symmetry. We emphasize again that this is a conditional result, though the condition is to be expected; see Remark~\ref{rmk:expectation} below.
\begin{theorem}[Main theorem on mass inflation]\label{mass.infl.thm}
Suppose the assumptions of Theorems \ref{final.blow.up.step} hold. \underline{Assume} in addition that 
\begin{equation}\label{add.assumption.0}
\int_{\mathcal H^+_1\cap\{v\geq 1\}} v^{6} (\rd_v\phi)^2(v)\, dv < +\infty
\end{equation}
and
\begin{equation}\label{add.assumption}
\int_{\mathcal H^+_1\cap\{v\geq 1\}} v^{8}(\rd_v^2\phi)^2(v)\, dv < +\infty. 
\end{equation}
Then the Hawking mass $m=+\infty$ on the Cauchy horizon $\CH_1$.
\end{theorem}

\begin{remark}[Alternative additional assumptions] 
The additional assumptions \eqref{add.assumption.0} and \eqref{add.assumption} are needed so that we can apply 
Theorem~\ref{thm:main} to the corresponding linear equations.

For this purpose, technically, we can alternatively assume
$$\int_{\mathcal H^+_1\cap\{v\geq 1\}} v^{6} (\rd_v\phi)^2(v)\, dv = +\infty$$
and
$$\int_{\mathcal H^+_1\cap\{v\geq 1\}} v^{6}(\rd_v^2\phi)^2(v)\, dv < +\infty.$$ 
This would still in principle be consistent with \cite{LO.interior, LO.exterior}. However, this is \underline{not} expected to hold; see Remark~\ref{rmk:expectation} immediately below.
\end{remark}

\begin{remark}[Expected behavior along $\EH_1$]\label{rmk:expectation}
According to the linear analysis in \cite{AAG,HintzPriceLaw}, for solutions $\phi$ to the linear wave equation on a fixed Reissner--Nordstr\"om spacetime, $|\rd_v\phi|\ls v^{-4}$ and $|\rd_v^2\phi| \ls v^{-5}$ along $\EH_1$. If these (or even slightly weaker versions) were also to hold for the \emph{nonlinear} solution, then we would have verified the assumptions \eqref{add.assumption.0} and \eqref{add.assumption}. 
\end{remark}

\subsection{Criterion for mass inflation}\label{sec:MI.crit} We begin the proof of Theorem~\ref{mass.infl.thm} by establishing a criterion for mass inflation in this subsection.

In order to make sense of our mass inflation criterion, we prove a simple lemma which follows from Theorem~\ref{main.theorem.C0.stability}.
\begin{lemma}\label{lem:cont.rduphi.limit}
For every $u <u_s$, the limit $\lim_{v\to +\infty}(\rd_u\phi)(u,v)$ exists. Moreover, $\lim_{v\to +\infty}(\rd_u\phi)(u,v)$ is a continuous function for $u \in (-\infty,u_s)$.
\end{lemma}
\begin{proof}
This follows easily from using the wave equation the wave equation $\rd_v (r\rd_u \phi) = -\rd_v r \rd_u \phi$, and controlling the terms on the right-hand side with the estimates in Theorem~\ref{main.theorem.C0.stability}. \qedhere
\end{proof}

We are now ready to state a criterion for mass inflation. For this purpose, it is more convenient to switch to the $(u,V)$ coordinate system as given in Section~\ref{sec:CH.blowup.from.old.paper}. In the $(u,V)$ coordinate system, Lemma~\ref{lem:cont.rduphi.limit} means that $\lim_{V\to 1}(\rd_u\phi)(u,V)$ is well-defined and continuous on $u \in (-\infty, u_s)$. In particular, the criteria in Proposition~\ref{basic.crit} make sense.
\begin{proposition}[Criterion for mass inflation]\label{basic.crit}
Suppose the assumptions of Theorem~\ref{thm.nonpert} hold. Then the following holds.
\begin{enumerate}
\item If $\sup_{V} |m|(u_*,V) <+\infty$, then there exists $u_{**} < \min\{ u_*, u_s\}$ such that $\lim_{V\to 1}(\rd_u\phi)(u,V)=0$ for every $u< u_{**}$.
\item If there exists a sequence $\{u_i\}_{i=1}^{+\infty}\subset (-\infty,u_s)$ such that $u_i \to -\infty$ and $\lim_{V\to 1}(\rd_u\phi)(u_i, V) \neq 0$, then the Hawking mass blows up identically on $\CH$, i.e.,~$\lim_{V\to 1} m(u,V) =+\infty$ for all $u \in (-\infty, u_{\CH_1})$.
\end{enumerate}
\end{proposition}
\begin{proof}
\pfstep{Step~1: Reduction to the renormalized Hawking mass} As a preliminary observation, note that the blow-up of the Hawking mass $m$ on $\CH_1$ is equivalent to the blow-up of the renormalized Hawking mass $\varpi$. This follows simply from the definition \eqref{eq:renormalized} and that $\bf e$ is a constant and $r$ has a non-zero limit on $\CH_1$ (see \cite[part (2)(e) of Theorem~4.1]{LO.interior}). We will in fact prove blow-up of $\varpi$, which is slightly more convenient because of the monotonicity properties that it enjoys.

\pfstep{Step~2: Signs of $\rd_u r$ and $\rd_V r$} Fix any $-\infty< \bar{u} < u_* < u_{\CH_1}$ with $\bar{u} < u_s$. We claim that after choosing $V_* \in (V_1, 1)$ sufficiently close to $1$, we have $\rd_u r(u,V), \, \rd_V r(u,V) < 0$ for $(u,V) \in [\bar{u}, u_*] \times [V_*, 1)$.

For $\rd_V r$, we use the following estimate from \cite[equation (10.13)]{LO.interior} 
\begin{equation}\label{eq:rdVr.uniform}
\sup_{(u, V) \in [\bar{u},u_*] \times [V_{\ast}, 1)} \Big| \rd_u (r \rd_V r) \Big|(u,V) \leq C.
\end{equation}
By \eqref{eq:nonpert-blowup-dvr} in Theorem~\ref{thm.nonpert} (and the continuous extendibility of $\log\Omg$), $\rd_V r\to -\infty$ for any fixed $u \in (-\infty,u_{\CH_1})$. Therefore, when combined with \eqref{eq:rdVr.uniform}, we see that there exists $V_* \in (V_1,1)$ such that 
\begin{equation}\label{eq:MI.rdVr}
\rd_V r(u,V)<0 \mbox{ whenever $(u,V) \in [\bar{u}, u_*] \times [V_*, 1)$}.
\end{equation}

Turning to $\rd_u r$, notice that $\rd_u r<0$ near $i^0_1$ by asymptotic flatness. Combining this with the Raychaudhuri equation (the second equation in \eqref{eqn.Ray}), we see that after choosing $V_*$ closer to $1$, 
\begin{equation}\label{eq:MI.rdur}
\rd_u r(u,V)<0 \mbox{ whenever $(u,V) \in (-\infty, u_{\CH_1}) \times [V_*, 1)$}.
\end{equation}

From now on, fix $V_*$.

\pfstep{Step~3: Monotonicity of the renormalized Hawking mass} Suppose $\sup_{V}|\varpi|(u_*,V) <\infty$ for some $u_* \in (-\infty,u_{\CH_{1}})$.
 
Using \eqref{eq:MI.rdVr} and the $\rd_u \varpi$ equation in \eqref{eq:varpi}, we obtain that $\rd_u\varpi(u,V)\geq 0$ whenever $(u,V) \in [\bar{u}, u_*] \times [V_*, 1)$. This monotonicity together with the assumption $\sup_{V}|\varpi|(u_*,V) <\infty$ imply that $\sup_{V}\varpi(u,V) <+ \infty$ for every $u \leq u_*$. 

On the other hand, by \eqref{eq:MI.rdur} and the $\rd_V \varpi$ equation in \eqref{eq:varpi}, $\varpi$ is monotonically increasing in $V$ for $V \in [V_*,1)$. In particular, $\varpi(u,V)$ is bounded below for $(u,V) \in (-\infty, u_{\CH_1}) \times [V_*, 1)$.

In particular, combining the upper and lower bounds, we deduce that $\sup_{V}|\varpi|(u,V) <\infty$ for all $u\leq u_*$.

\pfstep{Step~4: Proof of part (1)} Suppose $\sup_{V}|\varpi|(u_*,V) < +\infty$ for some $u_* \in (-\infty, u_{\CH_1})$. Taking $u_{**} < \min\{u_*, u_s\}$, we know by Step~3 that 
\begin{equation}\label{eq:MI.contradiction}
\sup_{V}|\varpi|(u_{**},V) < +\infty. 
\end{equation}
In order to prove part (1), our goal is to show that $\lim_{V\to 1}(\rd_u\phi)(u,V)=0$ for every $u< u_{**}$.

Assume for the sake of contradiction that there exists $u_0\in (-\infty,u_{**})$ such that $\lim_{V\to 1}(\rd_u\phi)(u_0,V)\neq 0$. Since $\lim_{V\to 1}(\rd_u\phi)(u,V)$ is a continuous function in $u$ (by Lemma~\ref{lem:cont.rduphi.limit}), there exist $u_1, u_2\in (-\infty,u_*)$ with $u_1<u_0<u_2$ such that $\lim_{v\to \infty}(\rd_u\phi)^2(u)\geq a>0$ for some $a>0$ whenever $u\in [u_1,u_2]$.

At the same time, we also have by \eqref{eq:nonpert-blowup-dvr} that $\lim_{V\to 1}\f{\rd_V r}{\Omg^2}(u,V)\to -\infty$ for every $u \in (-\infty, u_{\CH_1})$. 
Therefore, by Fatou's lemma (notice that $-\f{r\rd_V r}{\Omg^2}(\rd_u\phi)^2\geq 0$), we have
\begin{equation}\label{varpi.inf.int}
\liminf_{v\to \infty} \int_{u_1}^{u_2} (-\f{r\rd_V r}{\Omg^2}(\rd_u\phi)^2)(u',v)\, du'\geq  \int_{u_1}^{u_2} \liminf_{v\to \infty}(-\f{r\rd_V r}{\Omg^2}(\rd_u\phi)^2)(u',v)\, du'=\int_{u_1}^{u_2}a\cdot\infty\,du'=\infty.
\end{equation}

Now, we take an increasing sequence $\{V_n\}_{n=1}^{+\infty} \subset [V_*, 1)$ with $\lim_{n\to +\infty} V_n = 1$. Since $\varpi$ is increasing in $V$ for $V \in [V_*,1)$ (see Step~3), $\varpi(u, V_n)$ is an increasing sequence for every $u$. Integrating the $\rd_u \varpi$ equation in \eqref{eq:varpi} and using \eqref{eq:MI.rdVr}, we therefore obtain
\begin{equation*}
\begin{split}
\lim_{n\to\infty}\varpi(u_{**},V_n) =&\: \lim_{n\to\infty}\left(\varpi(u_1,V_n)+\int_{u_1}^{u_{**}} (-\f{r\rd_V r}{\Omg^2}(\rd_u\phi)^2)(u',V_n)\, du'\right)\\
\geq &\: \varpi(u_1,V_1)+\liminf_{n\to\infty}\int_{u_1}^{u_2} (-\f{r\rd_V r}{\Omg^2}(\rd_u\phi)^2)(u',V_n)\, du'=+\infty,
\end{split}
\end{equation*}
where in the last step we have used \eqref{varpi.inf.int}. However, this contradicts \eqref{eq:MI.contradiction} above! This concludes the proof of part (1).

\pfstep{Step~5: Proof of part (2)} By part (1), if the assumption of part (2) holds, then $\sup_V |m|(u,V) = +\infty$ for all $u \in (-\infty, u_{\CH_1})$. Using the fact that $m(u,V)$ is increasing in $V$ for sufficiently large $V$ (recall Step~3 above), it following that $\lim_{V\to 1} m(u,V) = +\infty$  for all $u \in (-\infty, u_{\CH_1})$, as desired. \qedhere
\end{proof}

\subsection{Estimates for the inhomogeneous wave equation on fixed Reissner--Nordstr\"om}\label{sec:EE.inhom}

Before we proceed to the proof of Theorem~\ref{mass.infl.thm} (see Section~\ref{sec:proofofmassinflation}), we study in this subsection the inhomogeneous wave equation on $(\mathcal M_{RN}, g_{RN})$:
$$\Box_{g_{RN}}\phi=F.$$
This will be useful for the perturbative argument in Section~\ref{sec:proofofmassinflation}.

\begin{theorem}\label{scat.SS.inho.u}
Let $p\in (1,\infty)$, $\sigma\in (1,+\infty)$ and $u_s\in (-\infty,-1)$.

Let $F:\mathcal M_{RN} \cap \{(u,v): u\in (-\infty,u_s),\,v\in [1,+\infty)\} \to \mathbb R$ be a spherically symmetric smooth function such that the following holds:
\begin{equation}\label{E.inho.def}
\begin{split}
\|F\|_{\mathcal E_{inho,p,\sigma}} \doteq \left(\int_{1}^{\infty}\int_{-\infty}^{u_s} |u|^p |F|^2(u,v)\,(\Omg_{RN}^4(1+|r^*|)^\sigma)(u,v)\,du\,dv \right)^{\f 12}<\infty.
\end{split}
\end{equation}
Suppose $\phi:\mathcal M_{RN} \cap \{(u,v): u\in (-\infty,u_s),\,v\in [1,+\infty)\} \to \mathbb R$ is a smooth spherically symmetric solution to
$$\Box_{g_{RN}}\phi = F$$
in $\mathcal M_{RN}\cap \{(u,v): u\in (-\infty,u_s), v\in (1,+\infty)\}$ such that $(\rd_v\phi)(u,v)$ extends smoothly to $\EH_1$ for all $u\in [1,+\infty)$ and\footnote{Note that it is not necessary to assume that $\rd_u\phi$ has a limit on $\CH_1$ for $u\in (-\infty,u_s)$. In fact, it is not difficult to prove using a approximation argument that such a limit exists. However, since in the applications we already know that such a limit exists, we will simply put this as an assumption of the theorem.} $\lim_{v\to +\infty}(\rd_u\phi)(u,v)$ exists for all $u\in (-\infty,u_s)$.

Then, there exists a constant $C>0$ depending only on $p$, $\sigma$, $u_s$ and the Reissner--Nordstr\"om parameters ${\bf e}$ and $M$ (but \underline{independent} of $\phi$) such that 
\begin{equation}\label{eq:scat.SS.inho.u}
\begin{split}
&\: \int_{-\infty}^{u_s} |u|^p \lim_{v\to +\infty}(\rd_u\phi)^2(u,v)\, du \\
 \leq &\: C\left(\int_1^{+\infty} v^p \lim_{u\to -\infty} (\rd_v\phi)^2(u,v)\, dv  + \int_{-\infty}^{u_s} |u|^p(\rd_u\phi)^2(u,v=1)\, du + \|F\|_{\mathcal E_{inho,p,\sigma}}^2\right).
\end{split}
\end{equation}
\end{theorem}
\begin{proof}
In the proof, unless explicitly stated otherwise, $C>0$, as well as implicit constants in $\ls$ and $\sim$, depend only on $p$, $\sigma$, $u_s$, ${\bf e}$ and $M$.

\pfstep{Step 1: Preliminary reductions} By Fatou's Lemma, it suffices to show that 
\begin{equation}\label{eq:reduction.by.Fatou}
\liminf_{v\to+\infty}\int_{\min\{-\f 12 v,u_s\}}^{u_s} |u|^p (\rd_u\phi)^2(u,v)\, du
\end{equation}
is bounded by the RHS of \eqref{eq:scat.SS.inho.u}.

This will be convenient in that we can apply the energy estimates in Lemma~\ref{lem:energy.identity} only to region of finite $v$.

\pfstep{Step~2: The curve $\gamma$ and partition of the spacetime} We will apply different estimates in two different regions of $\mathcal M_{RN} \cap \{(u,v): u\in (-\infty,u_s),\,v\in [1,+\infty)\}$. For this purpose, we divide the region into two. (Note a similar argument in \cite{D2,LS,DL}.)

Let $f(u,v)=u+v-\sqrt{v}$ for $u\in \mathbb R$ and $v\geq 1$. Introduce a hypersurface $\gamma\subset \{(u,v): u\in (-\infty,u_s),\,v\in [1,+\infty)\}$ by 
$$\gamma=f^{-1}(1).$$ 
Notice that $g^{-1}(df,df)=2 g^{-1}(du, (1-\f{1}{2\sqrt{v}})dv)<0$ when $v\geq 1$. Therefore $\gmm$ is a spacelike hypersurface.

We first propagate the estimates to $\gamma$ (Step~3) and then propagate the estimates from $\gamma$ to its future (Step~4) to obtain the final desired bounds.

\pfstep{Step~3: Past of $\gamma$} We now apply \eqref{EnergyEst} in the region 
$$\mathcal D_{Step\,3}  \doteq  \{(u,v): u\in (-\infty,u_s),\, v\in (1,+\infty),\, f(u,v) \in (-\infty,1) \}.$$
$\rd \mathcal D_{Step\,3}$ naturally decomposes into four components. On $\rd\mathcal D_{Step\,3}\cap \{u=u_s\}$, $\rd\mathcal D_{Step\,3}\cap \EH_1$ and $\rd\mathcal D_{Step\,3}\cap \{v=1\}$, we have the volume forms defined in Section~\ref{sec:vol}. On $\rd\mathcal D_{Step\,3}\cap\gamma$, we define a positive volume form $\vol_{\gamma}$ such that $\vol=df\wedge \vol_{\gamma}$.

We now take $Y$ to be
\begin{equation}\label{Y.in.step.3}
Y \doteq 2 y^N(r^*) \left(w_{p}(v)\rd_v+ w_{p}(-u)\rd_u\right).
\end{equation}
where $w_p$ is as in Definition~\ref{def:w.weight}
and $y$ is as in \eqref{y.def}. (Note that this is exactly \eqref{Y.def} with $p_1 = p$ and $p_2 = 0$.)

Using \eqref{EnergyEst}, we obtain\footnote{Note that since \eqref{EnergyEst} strictly speaking should be applied to compact domains, we need to carry out an approximation argument by considering an exhaustion of $\mathcal D_{Step\,3}$; we omit the details.} that for some $C>0$
\begin{equation}\label{eq:main.est.before.gamma}
\begin{split}
&\:\underbrace{\int_{\rd\mathcal D_{Step\,3}\cap\gamma} \mathbb T[\phi](-df^\sharp, Y)\,\vol_{\gamma}}_{\doteq I}+ \underbrace{\int_{\rd\mathcal D_{Step\,3}\cap \{u=u_s\}} \mathbb T[\phi](-du^\sharp, Y)\,\vol_{u}}_{\doteq II}+\underbrace{\int_{\mathcal D_{Step\,3}} \mathbb T[\phi]_{\mu\nu}{ }^{(Y)}\pi^{\mu\nu} \, \vol}_{Main\,good\,term}\\
\leq &\:\underbrace{\int_{\rd\mathcal D_{Step\,3}\cap \EH_1} \mathbb T[\phi](-du^\sharp, Y)\,\vol_{u}}_{\doteq III} + \underbrace{\int_{\rd\mathcal D_{Step\,3}\cap \{v=1\}} \mathbb T[\phi](-dv^\sharp, Y)\,\vol_{v}}_{\doteq IV} + C\underbrace{\int_{\mathcal D_{Step\,3}} |Y\phi| |F| \vol}_{Error\,term}.
\end{split}
\end{equation}
Note that $Y$, $-df^\sharp$, $-du^\sharp$ and $-dv^\sharp$ are all future directed causal and therefore by Proposition~\ref{positivity}, $I$, $II$, $III$ and $IV$ are all non-negative.

Moreover, a direct computation (using in particular Section~\ref{sec:vol}) shows that
\begin{equation}\label{past.of.gamma.data}
\begin{split}
&\:\int_{\rd\mathcal D_{Step\,3}\cap \EH_1} \mathbb T[\phi](-du^\sharp, Y)\,\vol_{u}  + \int_{\rd\mathcal D_{Step\,3}\cap \{v=1\}} \mathbb T[\phi](-dv^\sharp, Y)\,\vol_{v}\\
\sim & \int_1^{+\infty} v^p \lim_{u\to -\infty} (\rd_v\phi)^2(u,v)\, dv  + \int_{-\infty}^{u_s} |u|^p(\rd_u\phi)^2(u,v=1)\, du.
\end{split}
\end{equation}

To control the ``Error term'' above, we argue as follows:
\begin{align}
&\:\mbox{Error term}\nonumber\\
\ls &\:\int_{\mathcal D_{Step\,3}} (v^p|\rd_v\phi| + |u|^p|\rd_u\phi|) |F|(u,v)\,\Omg_{RN}^2(u,v)\,du\,dv \label{before.gamma.error.1}\\
\ls &\:\left(\int_{\mathcal D_{Step\,3}} (v^p|\rd_v\phi|^2 + |u|^p|\rd_u\phi|^2)(u,v)\,(1+|r^*(u,v)|)^{-\sigma}\,du\,dv\right)^{\f 12} \nonumber\\
&\: \times \left(\int_{\mathcal D_{Step\,3}} (v^p + |u|^p) |F|^2(u,v)\,(\Omg_{RN}^4(1+|r^*|)^\sigma)(u,v)\,du\,dv\right)^{\f 12} \label{before.gamma.error.2}\\
\ls &\: \left(\int_{\mathcal D_{Step\,3}} \mathbb T[\phi]_{\mu\nu}{ }^{(Y)}\pi^{\mu\nu} \,\vol \right)^{\f 12} \|F\|_{\mathcal E_{inho,p,\sigma}},\label{before.gamma.error.3} 
\end{align}
where in \eqref{before.gamma.error.1}, we used \eqref{vol.est}; in \eqref{before.gamma.error.2}, we applied the Cauchy--Schwarz inequality; in \eqref{before.gamma.error.3}, we used\footnote{Technically, \eqref{bulk.term.to.prove} is not directly applicable because there $p_2$ is assumed to be $> 1$ (whereas we have $p_2 = 0$ here). Nevertheless, one checks that $p_2>1$ is only used for the later parts of the proof but not for \eqref{bulk.term.to.prove}.} \eqref{bulk.term.to.prove}, \eqref{E.inho.def} and the fact that $v \ls |u|$ in $\mathcal D_{Step\,3}$.

Now, 
\begin{itemize}
\item plugging \eqref{past.of.gamma.data} and \eqref{before.gamma.error.3} into \eqref{eq:main.est.before.gamma}, 
\item applying the Young's inequality to \eqref{before.gamma.error.3} and absorbing the term $\int_{\mathcal D_{Step\,3}} \mathbb T[\phi]_{\mu\nu}{ }^{(Y)}\pi^{\mu\nu} \,\vol$ to the LHS by the ``Main bulk term'', and
\item dropping the good term $II$ on the LHS,
\end{itemize}
we obtain
\begin{equation}\label{past.of.gamma.main}
\begin{split}
&\int_{\rd\mathcal D_{Step\,3}\cap\gamma} \mathbb T[\phi](df^\sharp, Y)\,\vol_{\gamma}\leq \mbox{(RHS of \eqref{eq:scat.SS.inho.u})}.
\end{split}
\end{equation}

\pfstep{Step~4: Future of $\gamma$} Let $v_*\in [1,+\infty)$. We now consider the region to the future of $\gamma$, but to the past of some $\{v=v_*\}$. Precisely, we consider the (precompact) region
$$\mathcal D_{Step\,4}  \doteq  \{(u,v): u\in (-\infty,u_s),\, v\in (1,v_*),\, f(u,v) \in (1,+\infty) \}.$$

Define the vector field $Z$ by
$$ Z \doteq  y^N  ( \rd_v+|u|^{p}\rd_u),$$
where $y$ is again as in \eqref{y.def} and $N$ is a sufficiently large constant to be chosen below.

We now apply \eqref{EnergyEst} to obtain
\begin{equation}\label{main.Z}
\begin{split}
&\:\int_{\rd\mathcal D_{Step\,4}\cap\{v = v_*\}} \mathbb T[\phi](-dv^\sharp, Z)\, \vol_v+ \int_{\rd\mathcal D_{Step\,4}\cap\{u=u_s\}} \mathbb T[\phi](-du^\sharp, Z)\, \vol_u +\underbrace{\int_{\mathcal D_{Step\,4}} \mathbb T[\phi]_{\mu\nu}{ }^{(Z)}\pi^{\mu\nu} \, \vol}_{Main\,bulk\,term}\\
\leq &\:\underbrace{\int_{\rd\mathcal D_{Step\,4}\cap \gamma} \mathbb T[\phi](-df^\sharp, Z)\,\vol_{\gamma}}_{Data\,term}+C\underbrace{\int_{\mathcal D_{Step\,4}}|Z\phi| |F|\, \vol}_{Error\,term}.
\end{split}
\end{equation}
An explicit computation shows that
\begin{equation}\label{Z.1}
\int_{\rd\mathcal D_{Step\,4}\cap\{v = v_*\}} \mathbb T[\phi](-dv^\sharp, Z)\, \vol_v \sim \int_{\max\{1+\sqrt{v_*}-v_*,u_s\}}^{u_s} |u|^p(\rd_u\phi)^2(u,v_*)\, du.
\end{equation}

To handle the ``Data term'', note that there exists a constant $c>0$ such that $Y-cZ$ is future-directed and timelike (where $Y$ as in Step 2). Therefore, by Lemma~\ref{positivity}, $\mathbb T[\phi](df^\sharp, Y) - c\mathbb T[\phi](df^\sharp, Z) \geq 0$ pointwise. As a result, the ``Data term'' in \eqref{main.Z} is bounded (up to a constant factor) by \eqref{past.of.gamma.main}. In particular, we have
\begin{equation}\label{Z.2}
\begin{split}
\mbox{Data term}\leq C\mbox{(RHS of \eqref{eq:scat.SS.inho.u})}.
\end{split}
\end{equation}

For the ``Main bulk term'' in \eqref{main.Z}, we compute
\begin{equation}\label{Z.3}
\begin{split}
\mathbb T_{\mu\nu}{ }^{(Z)}\pi^{\mu\nu}=&-\f{4}{\Omg_{RN}^2}\left((\rd_uZ^v)(\rd_v\phi)^2+(\rd_v Z^u)(\rd_u\phi)^2\right)-\f{4}{r} (Z^v+Z^u)(\rd_u\phi\rd_v\phi)\\
=&\underbrace{\f{2 N(\sigma-1)}{\Omg_{RN}^2}y^{N-1}(1+|r^*|)^{-\sigma}\left((\rd_v\phi)^2+|u|^{p}(\rd_u\phi)^2\right)}_{\doteq I}-\underbrace{\f{4}{r} y^{N} (1+|u|^{p})(\rd_u\phi\rd_v\phi)}_{\doteq II}.
\end{split}
\end{equation}
In $\mathcal D_{Step\,4}$, we have $\Omg_{RN}^2(1+|r^*|)^{\sigma}\ls e^{-2\kappa_-(v+u)}(1+(v+u))^{\sigma}\ls (\sqrt{v})^\sigma e^{-2\kappa_-\sqrt{v}}\ls (\sqrt{|u|})^{\sigma} e^{-2\kappa_-\sqrt{|u|}}$. In particular, for $N$ sufficiently large (depending on $\sigma$, $p$, $M$ and ${\bf e}$), we can use the AM-GM inequality to show that the term $I$ dominates the term $II$ \eqref{Z.3}. Therefore, there exists a constant $c>0$ (depending on $\sigma$, $p$, $M$ and ${\bf e}$) such that
\begin{equation}\label{Z.4}
\mathbb T_{\mu\nu}{ }^{(Z)}\pi^{\mu\nu}\geq \f{c}{\Omg_{RN}^2(1+|r^*|)^{\sigma}}\left((\rd_v\phi)^2+|u|^{p}(\rd_u\phi)^2\right)\geq 0.
\end{equation}

We next control the ``Error term'' in \eqref{main.Z} using the ``Main bulk term'' (which is just shown above in \eqref{Z.4} to have a good sign). More precisely, by the Cauchy--Schwarz inequality and \eqref{E.inho.def}, we have
\begin{equation}\label{Z.5}
\begin{split}
\mbox{Error term}\ls &\int_{\mathcal D_{Step\,4}}\left(|\rd_v\phi|(u,v)+|u|^{p}|\rd_u\phi|(u,v)\right)|F|(u,v)\,\Omg_{RN}^2(u,v)\,du\,dv\\
\ls & \left(\int_{\mathcal D_{Step\,4}} \left(\f{|\rd_v\phi|^2(u,v)}{(1+|r^*|)^{\sigma}}+|u|^{p}\f{|\rd_u\phi|^2(u,v)}{(1+|r^*|)^{\sigma}}\right)\,du\,dv\right)^{\f 12}\\
&\qquad\times\left(\int_{\mathcal D_{Step\,4}} |u|^{p}|F|^2(u,v)\,(\Omg_{RN}^4(1+|r^*|)^{\sigma})(u,v)\,du\,dv\right)^{\f 12}\\
\leq &\: \left(\int_{\mathcal D_{Step\,4}} \left(\f{|\rd_v\phi|^2(u,v)}{(1+|r^*|)^{\sigma}}+|u|^{p}\f{|\rd_u\phi|^2(u,v)}{(1+|r^*|)^{\sigma}}\right)\,du\,dv\right)^{\f 12} \|F\|_{\mathcal E_{inho,p,\sigma}}.
\end{split}
\end{equation}

Combining \eqref{main.Z}, \eqref{Z.1}, \eqref{Z.2}, \eqref{Z.4} and \eqref{Z.5}, and dropping the (non-negative) second term on LHS of \eqref{main.Z}, we obtain
\begin{equation}\label{Z.final}
\begin{split}
&\:\int_{\max\{1+\sqrt{v_*}-v_*,u_s\}}^{u_s} |u|^p(\rd_u\phi)^2(u,v_*)\, du+\int_{\mathcal D_{Step\,4}}\left(\f{|\rd_v\phi|^2(u,v)}{(1+|r^*|)^{\sigma}}+|u|^{p}\f{|\rd_u\phi|^2(u,v)}{(1+|r^*|)^{\sigma}}\right)\,du\,dv\\
\ls &\:\mbox{(RHS of \eqref{eq:scat.SS.inho.u})}
+ \left(\int_{\mathcal D_{Step\,4}}\left(\f{|\rd_v\phi|^2(u,v)}{(1+|r^*|)^{\sigma}}+|u|^{p}\f{|\rd_u\phi|^2(u,v)}{(1+|r^*|)^{\sigma}}\right)\,du\,dv\right)^{\f 12} \|F\|_{\mathcal E_{inho,p,\sigma}}.
\end{split}
\end{equation}
After using Young's inequality and the positive bulk term on the LHS to control the last term on the RHS, we obtain
\begin{equation}\label{Z.final.final}
\int_{\max\{1+\sqrt{v_*}-v_*,u_s\}}^{u_s} |u|^p(\rd_u\phi)^2(u,v_*)\, du \ls \mbox{(RHS of \eqref{eq:scat.SS.inho.u})}.
\end{equation}
Finally, note that for $v_*$ sufficiently large, $1+\sqrt{v_*}-v_*\leq -\f{v_*}{2}$. It follows from \eqref{Z.final.final} that for all $v_*$ sufficiently large
\begin{equation}
\int_{\min\{-\f 12 v_*,u_s\}}^{u_s} |u|^p (\rd_u\phi)^2(u,v_*)\, du \ls \mbox{(RHS of \eqref{eq:scat.SS.inho.u})}.
\end{equation}
In particular, this gives the desired estimate for \eqref{eq:reduction.by.Fatou} in Step~1. We thus conclude the proof. \qedhere
\end{proof}

\subsection{Proof of Theorem \ref{mass.infl.thm}}\label{sec:proofofmassinflation}

In order to prove Theorem \ref{mass.infl.thm}, it suffices to show that the condition in Proposition \ref{basic.crit} holds. This will be carried out in two steps: First, we associate to $\phi$ a corresponding solution $\phi_{lin}$ to the \underline{linear} wave equation on a \underline{fixed} Reissner--Nordstr\"om spacetime. Applying Theorem~\ref{thm:main}, we will show that there is an $L^2$-average polynomial lower bound of $\rd_u\phi_{lin}$ along the Cauchy horizon as $u\to -\infty$. Second, we show --- using perturbative estimates --- that in fact such an $L^2$-average polynomial lower bound also holds for $\rd_u\phi$. In particular, this implies that $\rd_u\phi\restriction_{\CH_1}(u_n) \neq 0$ along some sequence $u_n\to -\infty$ so that the desired mass inflation result follows from Proposition~\ref{basic.crit}.

We begin with the first step. Suppose we are given $(\mathcal M, g, \phi, F)$ as in Theorem~\ref{mass.infl.thm}. By Theorem \ref{main.theorem.C0.stability}, we know that in the interior region near $\EH_1$, the solution converges to a Reissner--Nordstr\"om solution with parameters $M$ and ${\bf e}$. Fix these $M$ and ${\bf e}$. From now on, take $(\mathcal M_{RN}, g_{RN})$ to be the Reissner--Nordstr\"om interior corresponding to these parameter.

Let $\phi_{lin}$ be the solution to the linear wave equation 
\begin{equation}\label{eq:phi.lin.wave}
\Box_{g_{RN}}\phi_{lin}=0
\end{equation}
in the region $\{(u,v): u\leq u_s,\,v\geq 1\}\subset \mathcal M_{RN}$ (where $u_s\leq -1$ is as in Theorem \ref{main.theorem.C0.stability}) with initial data 
\begin{equation}\label{phi.lin.data}
\begin{cases}
\phi_{lin}\restriction_{\EH_1\cap \{v\geq 1\}} \doteq \phi\restriction_{\EH_1 \cap \{v\geq 1\}}\\
\phi_{lin}\restriction_{\{v=1\}\cap \{u\leq u_s\}} \doteq \phi\restriction_{\{v=1\}\cap \{u\leq u_s\}}.
\end{cases}
\end{equation}

We now apply Corollary~\ref{cor:main} to obtain the following result:
\begin{proposition}\label{lin.blow.up}
$\phi_{lin}$ defined as above satisfies 
\begin{equation}\label{eq:lin.blow.up.conclusion}
\int_{-\infty}^{u_s} u^8 \left(\lim_{v\to\infty} (\rd_u\phi_{lin})^2 (u,v)\right)\, du = +\infty.
\end{equation}
\end{proposition}
\begin{proof}
It suffices to apply Corollary~\ref{cor:main} to $\phi_{lin}$ with $p = 8$. Note that the last four assumptions of Theorem~\ref{thm:main} are indeed verified: 
\begin{enumerate}\setcounter{enumi}{2}
\item By Theorem~\ref{main.theorem.C0.stability}, $\lim_{v\to +\infty} \phi\restriction_{\EH_1}(v) = 0$. Hence, by \eqref{phi.lin.data}, we also have $\lim_{v\to +\infty} \phi_{lin}\restriction_{\EH_1}(v) = 0$.
\item By \eqref{add.assumption.0}, the bound $\int_1^{+\infty} (1+v^2) (T\phi_{lin}\restriction_{\EH_1}) \, dv <+\infty$ holds.
\item By \eqref{eq:mass.inflation.lower.bd} and \eqref{add.assumption.0}, $p = 8$ is the smallest even integer for which \eqref{linear.lower.bd} holds.
\item Using \eqref{add.assumption}, and taking $p=8$ as above, we obtain that \eqref{imp.decay} holds.
\end{enumerate}

Now that we have checked all the assumptions of Corollary~\ref{cor:main}, it follows that \eqref{eq:main.CH.conclu.prime} holds for $p =8$. Since $u_s \leq -1$, this implies \eqref{eq:lin.blow.up.conclusion}. \qedhere
\end{proof}

Next, we turn to the second step of the argument. For this, we prove a \emph{perturbative} statement showing that the nonlinear $\phi$ is well-approximated by $\phi_{lin}$ on $\CH_1$ as $u\to-\infty$ in the following sense:
\begin{proposition}\label{linear.diff}
Let $\phi$ be as in the assumptions of Theorem \ref{mass.infl.thm}, and $\phi_{lin}$ be the solution to the linear wave equation on $\mathcal M_{RN}$ as defined in \eqref{eq:phi.lin.wave} and \eqref{phi.lin.data} (where we have identified the region $\{(u,v): u\in (-\infty, u_s),\, v\in [1,+\infty)\}$ in the spacetime we are studying with the corresponding subset of $\mathcal M_{RN}$). The following estimate holds: 
$$\int_{-\infty}^{u_s} u^{8}\left(\lim_{v\to\infty} (\rd_u\phi-\rd_u\phi_{lin})^2 (u,v)\right)\, du < +\infty.$$
\end{proposition}
\begin{proof}
 A direct computation shows that the difference $\phi-\phi_{lin}$ satisfies the following inhomogeneous wave equation:
$$\rd_u\rd_v(\phi-\phi_{lin})=-\f{\rd_v r_{RN}}{r_{RN}}\rd_u(\phi-\phi_{lin})-\left(\f{\rd_v r}{r}-\f{\rd_v r_{RN}}{r_{RN}}\right)\rd_u\phi-\f{\rd_u r_{RN}}{r_{RN}}\rd_v(\phi-\phi_{lin})-\left(\f{\rd_u r}{r}-\f{\rd_u r_{RN}}{r_{RN}}\right)\rd_v\phi.$$
This implies that for $F \doteq \Box_{g_{RN}}(\phi-\phi_{lin})$, we have
$$|F|\lesssim \Omg^{-2}_{RN}\left(\left|\f{\rd_v r}{r}-\f{\rd_v r_{RN}}{r_{RN}}\right||\rd_u\phi|+\left|\f{\rd_u r}{r}-\f{\rd_u r_{RN}}{r_{RN}}\right||\rd_v\phi|\right).$$
In particular, according to the estimates proven in Theorem \ref{main.theorem.C0.stability}, the following estimates hold for any $\rho<3$:

\begin{enumerate}
\item 
For $u+v\leq 0$, we have\begin{equation}\label{eq:F.est.trivial.1}
|F|\lesssim_{\rho} \Omg^{-2}_{RN}( \Omg_{RN}^2 v^{-2\rho})=v^{-2\rho}.
\end{equation}
\item
For $u+v\geq 0$, we have
\begin{equation}\label{eq:F.est.trivial.2}
|F|\lesssim_{\rho} \Omg^{-2}_{RN} |u|^{-\rho}v^{-\rho}.
\end{equation}
\end{enumerate}

From now on, we fix some $\rho\in (\f {11}4, 3)$.

We are ready to apply Theorem \ref{scat.SS.inho.u}. It follows from Theorem \ref{scat.SS.inho.u} that in order to prove the desired conclusion of the proposition, it suffices to bound $F$ in the norm $\mathcal E_{inho,p,\sigma}$ with $p=8$ and some $\sigma >1$. We will choose $\sigma>1$ sufficiently close to $1$ such that $10+\sigma -4\rho <0$ and $2\rho-\sigma >1$ (which is possible since we have chosen $\rho>\f {11}4$).

To bound $F$, we use \eqref{eq:F.est.trivial.1} and \eqref{eq:F.est.trivial.2} to obtain
\begin{equation}\label{eq:mass.inflation.perturbative.F}
\begin{split}
\|F\|_{\mathcal E_{inho,8,\sigma}}
= &\: \left(\int_{-\infty}^{\infty}\int_{-\infty}^{\infty} |u|^8 |F|^2(u,v)\,(\Omg_{RN}^4(1+|r^*|)^\sigma)(u,v)\,du\,dv \right)^{\f 12}\\
\ls_\rho &\: \underbrace{\left(\int_{-\infty}^{u_s}\int_1^{-u} |u|^{8} v^{-4\rho}(\Omg_{RN}^4(1+|r^*|)^\sigma)(u,v)\,dv\,du \right)^{\f 12}}_{ \doteq I}\\
&\: +\underbrace{\left(\int_{-\infty}^{u_s}\int_{-u}^{+\infty} |u|^{8}|u|^{-2\rho}v^{-2\rho}\,(1+|r^*(u,v)|)^\sigma\,dv\,du \right)^{\f 12}}_{ \doteq II}.
\end{split}
\end{equation}
We now bound the terms $I$ and $II$ in \eqref{eq:mass.inflation.perturbative.F}. 

To handle $I$, recall that $r^*=v+u$, and that in the integration domain of $I$, $\Omg_{RN}\ls e^{\kappa_+ (v+u)}$.
Therefore, since $8+\sigma-4\rho<-2$, we have
\begin{equation}\label{eq:mass.inflation.perturbative.F.1}
\begin{split}
I\ls_{\rho} &\left(\int_{-\infty}^{u_s}\int_1^{-u} |u|^{8} v^{-4\rho}e^{4\kappa_+ (v+u)}(1+v+|u|)^\sigma\,dv\,du \right)^{\f 12}\\
\ls_{\rho} &\left(\int_{-\infty}^{u_s} |u|^{8+\sigma-4\rho} \,du \right)^{\f 12} < +\infty.
\end{split}
\end{equation}
In the derivation of \eqref{eq:mass.inflation.perturbative.F.1}, we have used that $v\leq |u|$ in the region $u+v\leq 0$, and thus an integration by parts argument gives
$$\int_1^{-u} v^{-4\rho}e^{4\kappa_+ (v+u)}(1+v+|u|)^\sigma\,dv \ls \abs{u}^{\sgm} \int_1^{-u} v^{-4\rho} e^{4\kappa_+ (v+u)} \,dv \ls_{\rho} |u|^{-4\rho+\sigma}.$$

For II, since $9+\sigma-4\rho<-1$, $-2\rho + \sigma <-1$,  $r^*=v+u$, and $|u|\leq v$ in the region $u+v\geq 0$, we have 
\begin{equation}\label{eq:mass.inflation.perturbative.F.2}
\begin{split}
II\ls_{\rho} & \left(\int_{-\infty}^{u_s}\int_{-u}^{+\infty} |u|^{8}|u|^{-2\rho}v^{-2\rho}\,(1+v+|u|)^\sigma\,dv\,du \right)^{\f 12}\\
\ls_{\rho}& \left(\int_{-\infty}^{u_s} (|u|^{8-2\rho} \int_{-u}^{+\infty} v^{-2\rho+\sigma} \,dv)\,du \right)^{\f 12} \ls_{\rho} \left(\int_{-\infty}^{u_s}|u|^{9+\sigma-4\rho} \,du \right)^{\f 12} < +\infty.
\end{split}
\end{equation}

Plugging \eqref{eq:mass.inflation.perturbative.F.1} and \eqref{eq:mass.inflation.perturbative.F.2} into \eqref{eq:mass.inflation.perturbative.F}, we obtain
$$\|F\|_{\mathcal E_{inho,8,\sigma}} <+\infty.$$
The estimate \eqref{eq:scat.SS.inho.u} in Theorem~\ref{scat.SS.inho.u} thus implies that
$$\int_{-\infty}^{u_s} u^{8}\left(\lim_{v\to\infty} (\rd_u\phi-\rd_u\phi_{lin})^2 (u,v)\right)\, du < +\infty, $$
which is what we wanted to prove. \qedhere
\end{proof}

We now finally conclude the proof of Theorem \ref{mass.infl.thm}:
\begin{proof}[Proof of Theorem \ref{mass.infl.thm}]
Combing Propositions \ref{lin.blow.up} and \ref{linear.diff} and using the triangle inequality, we immediately obtain 
$$\int_{-\infty}^{u_s} u^{8}\left(\lim_{v\to\infty} (\rd_u\phi)^2 (u,v)\right)\, du = +\infty.$$
In particular, there exists a sequence $u_n\to -\infty$ such that for every $n\in \mathbb N$,
$$\lim_{v\to\infty} (\rd_u\phi)^2 (u_n,v)\neq 0.$$
The desired blowup of the Hawking mass then follows from Theorem~\ref{basic.crit}.
\end{proof}

\bibliographystyle{DLplain}
\bibliography{massinflation}
\end{document}

%% file: mainthm.pdf_tex
\begingroup%
  \makeatletter%
  \providecommand\color[2][]{%
    \errmessage{(Inkscape) Color is used for the text in Inkscape, but the package 'color.sty' is not loaded}%
    \renewcommand\color[2][]{}%
  }%
  \providecommand\transparent[1]{%
    \errmessage{(Inkscape) Transparency is used (non-zero) for the text in Inkscape, but the package 'transparent.sty' is not loaded}%
    \renewcommand\transparent[1]{}%
  }%
  \providecommand\rotatebox[2]{#2}%
  \ifx\svgwidth\undefined%
    \setlength{\unitlength}{520bp}%
    \ifx\svgscale\undefined%
      \relax%
    \else%
      \setlength{\unitlength}{\unitlength * \real{\svgscale}}%
    \fi%
  \else%
    \setlength{\unitlength}{\svgwidth}%
  \fi%
  \global\let\svgwidth\undefined%
  \global\let\svgscale\undefined%
  \makeatother%
  \begin{picture}(1,0.61538462)%
    \put(0,0){\includegraphics[width=\unitlength]{mainthm.pdf}}%
    \put(0.86180978,0.28038405){\color[rgb]{0,0,0}\makebox(0,0)[lt]{\begin{minipage}{0.16773306\unitlength}\raggedright $\mathcal{I}^{+}$\end{minipage}}}%
    \put(0.60879122,0.48909855){\color[rgb]{0,0,0}\makebox(0,0)[lb]{\smash{$\mathcal{CH}^{+}$}}}%
    \put(0.74769231,0.37986778){\color[rgb]{0,0,0}\makebox(0,0)[lb]{\smash{$i^{+}$}}}%
    \put(0.98065937,0.14886568){\color[rgb]{0,0,0}\makebox(0,0)[lb]{\smash{$i^0$}}}%
    \put(0.58153846,0.26448317){\color[rgb]{0,0,0}\makebox(0,0)[lb]{\smash{$\mathcal{H}^+$}}}%
    \put(0.27912088,0.16043964){\color[rgb]{0,0,0}\makebox(0,0)[lb]{\smash{$\Sigma_0$}}}%
    \put(0.23208792,0.37963491){\color[rgb]{0,0,0}\makebox(0,0)[lb]{\smash{$i^+$}}}%
    \put(0.09846154,0.26448317){\color[rgb]{0,0,0}\makebox(0,0)[lb]{\smash{$\mathcal{I}^+$}}}%
    \put(0.00483516,0.14886568){\color[rgb]{0,0,0}\makebox(0,0)[lb]{\smash{$i^0$}}}%
    \put(0.36307692,0.26448317){\color[rgb]{0,0,0}\makebox(0,0)[lb]{\smash{$\mathcal{H}^+$}}}%
    \put(0.31164835,0.48909855){\color[rgb]{0,0,0}\makebox(0,0)[lb]{\smash{$\mathcal{CH}^+$}}}%
  \end{picture}%
\endgroup%

%% file: mainstructure-new.pdf_tex
\begingroup%
  \makeatletter%
  \providecommand\color[2][]{%
    \errmessage{(Inkscape) Color is used for the text in Inkscape, but the package 'color.sty' is not loaded}%
    \renewcommand\color[2][]{}%
  }%
  \providecommand\transparent[1]{%
    \errmessage{(Inkscape) Transparency is used (non-zero) for the text in Inkscape, but the package 'transparent.sty' is not loaded}%
    \renewcommand\transparent[1]{}%
  }%
  \providecommand\rotatebox[2]{#2}%
  \ifx\svgwidth\undefined%
    \setlength{\unitlength}{1120bp}%
    \ifx\svgscale\undefined%
      \relax%
    \else%
      \setlength{\unitlength}{\unitlength * \real{\svgscale}}%
    \fi%
  \else%
    \setlength{\unitlength}{\svgwidth}%
  \fi%
  \global\let\svgwidth\undefined%
  \global\let\svgscale\undefined%
  \makeatother%
  \begin{picture}(1,0.28571429)%
    \put(0,0){\includegraphics[width=\unitlength,page=1]{mainstructure-new.pdf}}%
    \put(0.39441168,0.13023344){\color[rgb]{0,0,0}\makebox(0,0)[lt]{\begin{minipage}{0.10168679\unitlength}\raggedright $\mathcal{I}_{1}^{+}$\end{minipage}}}%
    \put(0.3431454,0.17147988){\color[rgb]{0,0,0}\makebox(0,0)[lb]{\smash{$i_{1}^{+}$}}}%
    \put(0.44878331,0.06962472){\color[rgb]{0,0,0}\makebox(0,0)[lb]{\smash{$i_{1}^0$}}}%
    \put(0.29714286,0.10454171){\color[rgb]{0,0,0}\makebox(0,0)[lb]{\smash{$\mathcal{H}_{1}^+$}}}%
    \put(0.00296051,0.06622519){\color[rgb]{0,0,0}\makebox(0,0)[lb]{\smash{$i_{2}^0$}}}%
    \put(0.10801971,0.17284184){\color[rgb]{0,0,0}\makebox(0,0)[lb]{\smash{$i_{2}^{+}$}}}%
    \put(0.12794084,0.10597028){\color[rgb]{0,0,0}\makebox(0,0)[lb]{\smash{$\mathcal{H}_{2}^+$}}}%
    \put(0.30694907,0.21624085){\color[rgb]{0,0,0}\makebox(0,0)[lt]{\begin{minipage}{0.11106681\unitlength}\raggedright $\mathcal{CH}_{1}^{+}$\\ \end{minipage}}}%
    \put(0.11982102,0.21762592){\color[rgb]{0,0,0}\makebox(0,0)[lt]{\begin{minipage}{0.11106681\unitlength}\raggedright $\mathcal{CH}_{2}^{+}$\\ \end{minipage}}}%
    \put(0.20635973,0.2334174){\color[rgb]{0,0,0}\makebox(0,0)[lb]{\smash{}}}%
    \put(0.0394077,0.13023344){\color[rgb]{0,0,0}\makebox(0,0)[lt]{\begin{minipage}{0.09374991\unitlength}\raggedright $\mathcal{I}_{2}^{+}$\end{minipage}}}%
    \put(0.19408316,0.0715982){\color[rgb]{0,0,0}\makebox(0,0)[lb]{\smash{$\Sigma_{0}$}}}%
    \put(0.06695872,0.09553957){\color[rgb]{0,0,0}\makebox(0,0)[lt]{\begin{minipage}{0.09374991\unitlength}\raggedright $\Ext_{2}$\end{minipage}}}%
    \put(0.34961178,0.0996212){\color[rgb]{0,0,0}\makebox(0,0)[lt]{\begin{minipage}{0.09374991\unitlength}\raggedright $\Ext_{1}$\end{minipage}}}%
    \put(0.15540816,0.140256){\color[rgb]{0,0,0}\makebox(0,0)[lb]{\smash{$\Int_{i^+_2}$}}}%
    \put(0.27836735,0.140256){\color[rgb]{0,0,0}\makebox(0,0)[lb]{\smash{$\Int_{i^+_1}$}}}%
    \put(0,0){\includegraphics[width=\unitlength,page=2]{mainstructure-new.pdf}}%
    \put(0.2067012,0.18469896){\color[rgb]{0,0,0}\makebox(0,0)[lb]{\smash{$\Int_{nonpert}$}}}%
    \put(0,0){\includegraphics[width=\unitlength,page=3]{mainstructure-new.pdf}}%
    \put(0.92657252,0.13023343){\color[rgb]{0,0,0}\makebox(0,0)[lt]{\begin{minipage}{0.10168679\unitlength}\raggedright $\mathcal{I}_{1}^{+}$\end{minipage}}}%
    \put(0.87530622,0.17147988){\color[rgb]{0,0,0}\makebox(0,0)[lb]{\smash{$i_{1}^{+}$}}}%
    \put(0.98094413,0.06962472){\color[rgb]{0,0,0}\makebox(0,0)[lb]{\smash{$i_{1}^0$}}}%
    \put(0.82930368,0.10454171){\color[rgb]{0,0,0}\makebox(0,0)[lb]{\smash{$\mathcal{H}_{1}^+$}}}%
    \put(0.53512137,0.06622519){\color[rgb]{0,0,0}\makebox(0,0)[lb]{\smash{$i_{2}^0$}}}%
    \put(0.64018058,0.17284184){\color[rgb]{0,0,0}\makebox(0,0)[lb]{\smash{$i_{2}^{+}$}}}%
    \put(0.66010167,0.10597028){\color[rgb]{0,0,0}\makebox(0,0)[lb]{\smash{$\mathcal{H}_{2}^+$}}}%
    \put(0.8391099,0.21624084){\color[rgb]{0,0,0}\makebox(0,0)[lt]{\begin{minipage}{0.11106681\unitlength}\raggedright $\mathcal{CH}_{1}^{+}$\\ \end{minipage}}}%
    \put(0.65198185,0.2176259){\color[rgb]{0,0,0}\makebox(0,0)[lt]{\begin{minipage}{0.11106681\unitlength}\raggedright $\mathcal{CH}_{2}^{+}$\\ \end{minipage}}}%
    \put(0.66709202,0.2334174){\color[rgb]{0,0,0}\makebox(0,0)[lb]{\smash{}}}%
    \put(0.73479914,0.24367377){\color[rgb]{0,0,0}\makebox(0,0)[lt]{\begin{minipage}{0.11106681\unitlength}\raggedright $\mathcal{S}$\end{minipage}}}%
    \put(0.57156854,0.13023343){\color[rgb]{0,0,0}\makebox(0,0)[lt]{\begin{minipage}{0.09374991\unitlength}\raggedright $\mathcal{I}_{2}^{+}$\end{minipage}}}%
    \put(0.72624398,0.0715982){\color[rgb]{0,0,0}\makebox(0,0)[lb]{\smash{$\Sigma_{0}$}}}%
    \put(0.59911955,0.09553955){\color[rgb]{0,0,0}\makebox(0,0)[lt]{\begin{minipage}{0.09374991\unitlength}\raggedright $\Ext_{2}$\end{minipage}}}%
    \put(0.88177262,0.09962119){\color[rgb]{0,0,0}\makebox(0,0)[lt]{\begin{minipage}{0.09374991\unitlength}\raggedright $\Ext_{1}$\end{minipage}}}%
    \put(0.68756901,0.140256){\color[rgb]{0,0,0}\makebox(0,0)[lb]{\smash{$\Int_{i^+_2}$}}}%
    \put(0.81052822,0.140256){\color[rgb]{0,0,0}\makebox(0,0)[lb]{\smash{$\Int_{i^+_1}$}}}%
    \put(0,0){\includegraphics[width=\unitlength,page=4]{mainstructure-new.pdf}}%
    \put(0.74314776,0.17469896){\color[rgb]{0,0,0}\makebox(0,0)[lb]{\smash{$\Int_{nonpert}$}}}%
  \end{picture}%
\endgroup%